\documentclass[12pt, draftcls, onecolumn]{IEEEtran}

\usepackage{amsfonts,stmaryrd,acronym,amssymb,amsthm,dsfont}
\usepackage[cmex10]{amsmath} 
\usepackage{graphicx}
\usepackage{url}
\usepackage{cite}
\usepackage{hyperref}
\usepackage{flushend}
\usepackage{enumerate}
\usepackage{lipsum}

\newtheorem{theorem}{Theorem}
\newtheorem{corollary}{Corollary}
\newtheorem{definition}{Definition}

\newtheorem{remark}{Remark}
\newtheorem{lemma}{Lemma}

\makeatletter
\def\blfootnote{\xdef\@thefnmark{}\@footnotetext}
\makeatother
\newcommand{\indic}[1]{\ensuremath{\mathds{1}}}

\newcounter{mytempeqcounter}

\allowdisplaybreaks
\allowbreak

\makeatletter%
\if@twocolumn%
\newcommand{\Figwidth}{\columnwidth}%

\def\twocolbreak{\nonumber\\ &}%
\def\twocolbreakquad{\nonumber\\ &\quad\quad{}}%
\def\twocolbreakonequad{\nonumber\\ &\quad{}}%
\def\twocolbreaktimes{\nonumber\\ &\times}%
\def\twocolnewline{\nonumber\\}%
\def\twocolAlignMarker{&}%
\def\onecolAlignMarker{}%
\else
\newcommand{\Figwidth}{4.5in}%

\def\twocolbreak{}%
\def\twocolbreakquad{}%
\def\twocolbreakonequad{}%
\def\twocolbreaktimes{}%
\def\twocolnewline{}%
\def\twocolAlignMarker{}%
\def\onecolAlignMarker{&}%
\fi%
\makeatother%

\def\Prob{\mathbb P}
\def\mi{\mathbb I}
\def\ent{\mathbb H}

\def\hatlw{\hat{w}}
\def\dhatlw{\check{w}}
\def\hatuw{\hat{W}}
\def\dhatuw{\check{W}}

\def\hatlu{\hat{u}}
\def\dhatlu{\check{u}}

\def\hatlv{\hat{v}}
\def\dhatlv{\check{v}}

\def\hatll{\hat{\ell}}
\def\dhatll{\check{\ell}}

\begin{document}

\title{Two-Multicast Channel with Confidential Messages}

\author{Hassan ZivariFard, Matthieu Bloch, {\em Senior Member, IEEE}, and Aria Nosratinia, {\em Fellow, IEEE}\thanks{H. ZivariFard and A. Nosratinia are with Department of Electrical Engineering, The University of Texas at Dallas, Richardson, TX, USA. M. Bloch is with School of Electrical and Computer Engineering, Georgia Institute of Technology, Atlanta, GA, USA. E-mail: hassan@utdallas.edu, matthieu.bloch@ece.gatech.edu, aria@utdallas.edu.}\thanks{The material in this paper was presented in part at the 55th Annual Allerton Conference on Communication, Control, and Computing, Monticello, IL, September 2017.}
}

\maketitle

\begin{abstract}
Motivated in part by the problem of secure multicast distributed storage, we analyze secrecy rates for a channel in which two transmitters simultaneously multicast to two receivers in the presence of an eavesdropper. Achievable rates are calculated via extensions of a technique due to Chia and El~Gamal and the method of output statistics of random binning. 
Outer bounds are derived for both the degraded and non-degraded versions of the channel, and examples are provided in which the inner and outer bounds meet. The inner bounds recover known results for the multiple-access wiretap channel, broadcast channel with confidential messages, and the compound MAC channel. An auxiliary result is also produced that derives an inner bound on the minimal randomness necessary to achieve secrecy in multiple-access wiretap channels.
\end{abstract}

\section{Introduction}

We study the multiuser secure multicast problem (Fig.~\ref{fig}), more specifically, when two transmitters multicast messages securely to two receivers in the presence of an eavesdropper. All senders, receivers, and eavesdropper are at different terminals. This problem is motivated in part by secure access of multiple users to data in a distributed cache, which is a multi-transmitter (multiple-access) multicast scenario~\cite{SecureMulticast,TwoMulticast}. This problem is also equivalent to a compound two-state multiple-access wiretap channel. It has been known~\cite{CompoundWT} that problems involving compound channels have an equivalent multicast representation, in which the channel to each multicast receiver is equivalent to one of the states of the compound channel.\footnote{The problem studied herein is the secrecy counterpart of the classical problem posed by Ahlswede~\protect\cite{Ahlswede}, which proved highly influential for the MAC channel~\protect\cite{Verdu50} and the interference channel~\protect\cite{HanKobayashi}.}

This paper takes a two-pronged approach to the analysis of the network mentioned above, producing a number of new results and insights. In Section~\ref{AchivRateRegionWeak} we present an analysis inspired by the work of Chia and El~Gamal~\cite{ChiaElGamal}, which uses Marton coding and indirect decoding (also known as non-unique decoding)~\cite{NairElGamal} to achieve an improved secrecy rate for the transmission of {\em one} common message to two receivers that may experience different channel statistics. 
In extending the method of Chia and El~Gamal to multiple transmitters, we introduce a two-level Marton-type coding with associated non-unique decoding.

\begin{figure}
\centering
\includegraphics[width=\Figwidth]{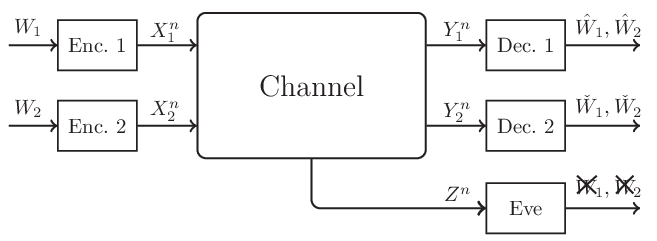}
\caption{Two-sender two-receiver channel with an eavesdropper}

\label{fig}
\end{figure}

In Section~\ref{AchivRateRegionStrong}, we employ the method of output statistics of random binning (OSRB)~\cite{OSRB} for analyzing the two-transmitter two-receiver problem (see also~\cite{RenesRenner2011} for a related approach). OSRB analyzes channel coding problems by conversion to a related source coding problem, where it tests achievability by probability approximation rather than counting arguments on typical sets, followed by a reverse conversion to complete the analysis. OSRB is well suited for secrecy problems because secrecy is tightly related to probability approximation. OSRB encoding is purely by random binning and is enabled by (and named after) the following asymptotic result: apply two independent random binning schemes on the same set and take a random sample from the set. The two bin indices corresponding to the random sample are statistically independent as long as binning rates are sufficiently small \cite{AlmostIndep, OSRB, RenesRenner2011}. 
We extend the tools and techniques of OSRB to match the requirements of the two-transmitter multicast problem.


The different parts of this paper complement each other, producing a more complete picture in the understanding of the problem of multi-transmitter secure multicast. 
The extension of the method of Chia and El~Gamal is utilized to highlight the minimal amount of randomness required to achieve secrecy rates over the multiple-access wiretap channel, and that therein channel prefixing can be replaced with superposition, in a  manner reminiscent of Watanabe and Oohama~\cite{OptimalRandomness} for minimizing the randomness resources for secrecy encoding.
The analysis based on OSRB generates strong secrecy, which interestingly has an expression that is a superset of the {\em achievable} weak secrecy calculated in the first part. Furthermore, the expression for the strong secrecy region can be greatly simplified via a constraint found in the weak secrecy analysis, highlighting the synergy between the two. More broadly, the developments in these two parts each offer techniques and insights that can potentially be useful in a wider class of problems.

Outer bounds for degraded and non-degraded channels are derived and shown to be tight against inner bounds in some special cases. 
Part of the results, including the proof of the outer bounds, appeared in an earlier conference version of this paper~\cite{AllertonPaper} and are not duplicated here in the interest of brevity.

A brief outline of the related literature is as follows. Multicasting with common information in the presence of an eavesdropper has been studied in~\cite{ForenSaleh,EkremUlukus13}, deriving inner bounds on secrecy capacity, and in some special cases also deriving secrecy capacity region. Salehkalaibar \emph{et al.}~\cite{ForenSaleh} studied a one-receiver, two-eavesdropper broadcast channel with three degraded message sets. Ekrem and Ulukus~\cite{EkremUlukus13} studied the transmission of public and confidential messages to two legitimate users, in the presence of an eavesdropper. Benammar and Piantanida~\cite{BenaPiantaWBC} calculated the secrecy capacity region of some classes of wiretap broadcast channels.

The MAC wiretap channel has been investigated in~\cite{LiangPoor,GMAWC,GMAWCJamming,YassaeeMAWC,WieseBoche,PierrotBloch2011,ForenPaper,Chou2018}. In~\cite{LiangPoor}, a discrete memoryless MAC with confidential messages has been studied that consists of a MAC with generalized feedback~\cite{Carleial} where each user's message must be kept confidential from the other. The multiple access wiretap channel~\cite{GMAWC,GMAWCJamming,ForenPaper} consists of a MAC with an additional channel output to an eavesdropper. In~\cite{GMAWC,GMAWCJamming}, achievable rate regions for the secrecy capacity region have been derived. Secrecy in the interference channel and broadcast channel has been studied in~\cite{LiuMaric}, where inner and outer bounds for the broadcast channel with confidential messages and the interference channel with confidential messages have been compared.


\section{Preliminaries}

Throughout this paper, random variables are denoted by capital letters and their realizations by lower case letters.  The set of $\epsilon-$strongly jointly typical sequences of length $n$, according to $p_{X,Y}$, is denoted by $\mathcal{T}_{\epsilon}^{(n)}({p_{X,Y}})$. For convenience in notation, whenever there is no danger of confusion, typicality will reference the random variables rather than the distribution, e.g., $\mathcal{T}_{\epsilon}^{(n)}(X,Y)$. The set of sequences $\{x^n: (x^n,y^n)\in T_\epsilon^{(n)}(X,Y)\}$ for a fixed $y^n$, when the fixed sequence $y^n$ is clear from the context, is denoted with the shorthand notation $\mathcal{T}_{\epsilon}^{(n)}(X|Y)$.  Superscripts denote the dimension of a vector, e.g., $X^n$. The integer set $\{1,\dots,M\}$ is denoted by $\llbracket 1,M\rrbracket$, and $X_{[i:j]}$ indicates the set $\{X_i,X_{i+1},\dots,X_j\}$. The cardinality of a set is denoted by $|\cdot|$. We utilize the total variation between probability mass functions (pmfs), defined by $||q-p||_1=\frac{1}{2}\sum_x |p-q|$. 
Following Cuff~\cite{Cuff} we use the concept of random pmfs denoted by capital letters (e.g. $P_X$).

\begin{definition}
\label{defi1}
A $(M_{1,n},M_{2,n},n)$ code for the considered model (Fig.~\ref{fig}) consists of the following:
\begin{enumerate}[i)]
\item Two message sets $\mathcal{W}_i=\llbracket 1,M_{i,n}\rrbracket$, $i=1,2$, from which independent messages $W_1$ and $W_2$ are drawn uniformly distributed over their respective sets.

\item Stochastic encoders $f_i$, $i=1,2$, which are specified by conditional probability matrices $f_i(X_i^n|w_i)$, where $X_i^n\in\mathcal{X}_i^n$, $w_i\in\mathcal{W}_i$ are channel inputs and private messages, respectively, and $\sum\nolimits_{x_i^n}{f_i(x_i^n|w_i)} = 1$. Here, $f_i(x_i^n|w_i)$ is the probability of the encoder producing the codeword $x_i^n$ for the message $w_i$.

\item A decoding function $\phi_1:\mathcal Y_1^n\to\mathcal{W}_1\times\mathcal{W}_2$ that assigns $(\hatlw_1,\hatlw_2) \in \llbracket 1,{M_{1,n}}\rrbracket\times\llbracket 1,M_{2,n}\rrbracket$ to received sequence $y_1^n$.

\item A decoding function $\phi_2:\mathcal Y_2^n\to\mathcal{W}_1\times\mathcal{W}_2$ that assigns $(\dhatlw_1,\dhatlw_2) \in \llbracket 1,{M_{1,n}}\rrbracket\times\llbracket 1,M_{2,n}\rrbracket$ to received sequence $y_2^n$.
\end{enumerate}
\end{definition}
The probability of error is given by:
\begin{equation*}
\label{pen}
P_{e} \triangleq\Prob\big(\{(\hatuw_1,\hatuw_2)\ne(W_1,W_2)\} \cup \{(\dhatuw_1,\dhatuw_2) \ne (W_1,W_2) \}\big).
\end{equation*}
\begin{definition}[\cite{BlochBarros}]
\label{defiperfect}
A rate pair $(R_1,R_2)$ is said to be achievable if there exists a sequence of $({M_{1,n}},{M_{2,n}},n)$ codes with ${M_{1,n}}\ge{2^{n{R_1}}},{M_{2,n}}\ge{2^{n{R_2}}}$, so that $P_e
\underset{n\rightarrow\infty}{\xrightarrow{\hspace{0.2in}}} 0$ and
\begin{align}
\frac{1}{n}&\mi(W_1,W_2;Z^n) \underset{n\rightarrow\infty}{\xrightarrow{\hspace{0.2in}}} 0\quad\mbox{for weak secrecy regime},\label{Secrecy_Defi}\\
&\mi(W_1,W_2;Z^n) \underset{n\rightarrow\infty}{\xrightarrow{\hspace{0.2in}}} 0\quad\mbox{for strong secrecy regime}.\label{Strong_Secrecy_Defi}
\end{align}
\end{definition}
\begin{definition}
\label{pmfDefi}
$p_X \mathop \approx q_X$ indicates ${\left\| {p_X - q_X} \right\|_1} < \epsilon$. For two random pmfs~\cite{Cuff}, $P_X \mathop
\approx Q_X$ indicates $\mathbb{E}{\left\|
P_X - Q_X \right\|_1} < \epsilon$.
\end{definition}

\section{Achievable Rate Region under Weak Secrecy}
\label{AchivRateRegionWeak}

We start with a lemma that fits Marton coding with indirect decoding in a MAC structure and produces an entropy bound needed in the secrecy analysis. Its basic idea can be highlighted as follows: given $X^n$, if we {\em independently} produce $2^{nR}$ random codevectors $Y^n$, we will have approximately $2^{nR-\mi(X^n;Y^n)}$ {\em jointly} typical pairs, i.e., the ``excess'' rate will determine the number of jointly typical pairs. This lemma extends the basic idea of excess rate to multiple codebooks, multiple conditioning, and furthermore, a generalization is made from a counting argument to the entropy of the index of the codebook, which is essential for the subsequent secrecy analysis.

\begin{lemma}
\label{lemma1}
\sloppy Consider random variables $(Q,U_0,V_0,U_1,V_1,Z)$ distributed according to $p_Qp_{U_0,U_1|Q}p_{V_0,V_1|Q}p_{Z|U_0,U_1,V_0,V_1}$. 
Draw random sequences $Q^n,U_0^n,V_0^n$ according to $\prod_{i = 1}^n p_Q(q_i)\; p_{U_0|Q}(u_{0,i}|q_i)\; p_{V_0|Q}(v_{0,i}|q_i)$. Conditioned on $U_0^n$, draw $2^{nS}$ i.i.d.\ copies of $U_1^n$ according to $\prod\nolimits_{i = 1}^n p_{U_1|U_0}(u_{1,i}|u_{0,i})$, denoted $U_1^n(\ell), \ell \in \llbracket1,2^{nS}\rrbracket$. Similarly, conditioned on $V_0^n$, draw $2^{nT}$ i.i.d.\ copies of $V_1^n$ according to $\prod\nolimits_{i = 1}^n p_{V_1|V_0}(v_{1,i}|v_{0,i})$, denoted $V_1^n(k), k \in \llbracket1,2^{nT}\rrbracket$. Let $L \in \llbracket1,2^{nS}\rrbracket$ and $K \in\llbracket1,2^{nT}\rrbracket$ be random variables with arbitrary pmf. If
\begin{align*}
   S&>\mi(U_1;Z|Q,U_0,V_0)+\delta_1(\epsilon)\\
   T&>\mi(V_1;Z|Q,U_0,V_0)+\delta_1(\epsilon)\\
   S+T&>\mi(U_1,V_1;Z|Q,U_0,V_0)+\delta_1(\epsilon)
\end{align*}
for a positive $\delta_1(\epsilon)$ and if for an arbitrary sequence $Z^n$, \[\Prob\big(
     (Q^n,U_0^n,V_0^n,U_1^n(L),V_1^n(K),Z^n) \in
     \mathcal{T}_{\epsilon}^{(n)}\big) \underset{n\rightarrow\infty}{\xrightarrow{\hspace{0.2in}}} 1.
     \]
     there
     exists a positive $\delta_2(\epsilon) \underset{\epsilon\rightarrow 0}{\xrightarrow{\hspace{0.2in}}} 0$, such that for $n$ sufficiently large
\begin{align}
\label{exlemma}
&\ent(L,K|Q^n,U_0^n,V_0^n,Z^n,{\mathcal C}) \twocolbreak\le
n(S+T-\mi(U_1,V_1;Z|Q,U_0,V_0))  + n \delta_2(\epsilon).
\end{align}
where ${\mathcal C} = \{U_1^n(1), \ldots, U_1^n(2^{nS}), V_1^n(1), \ldots,V_1^n(2^{nT}) \}$.

\end{lemma} 
\begin{figure}
\centering
\includegraphics[width=\Figwidth]{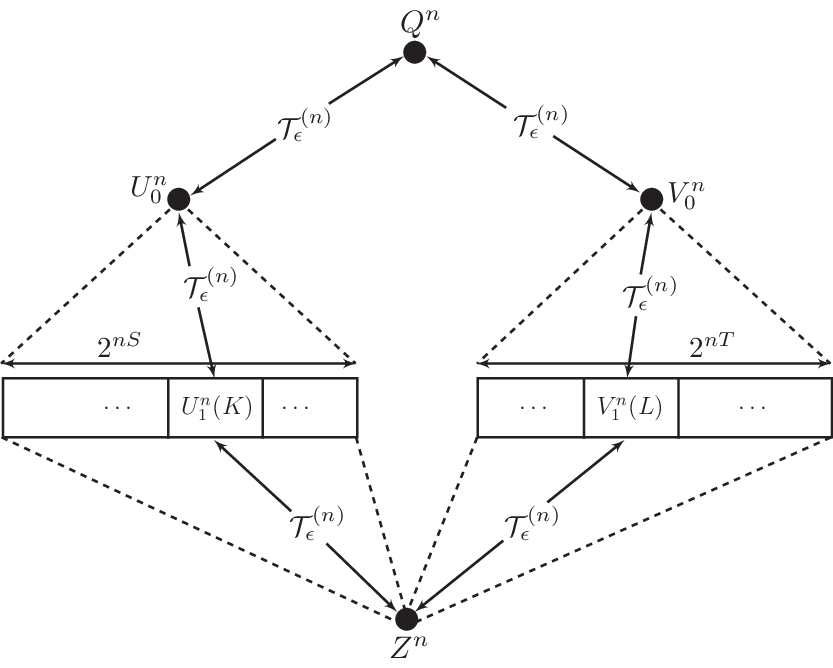}
\caption{Structure of Lemma~\ref{lemma1}: subject to jointly typical sequences $(Q^n,U_0^n,V_0^n,U_1^n(K),V_1^n(L),Z^n)$, finding a bound on the conditional entropy of $(K,L)$, thus implicitly bounding the number of sequence pairs that can be jointly typical with $(Q^n,Z^n)$ from codebooks with certain size.}

\label{LemmaFig}
\end{figure}

The proof is provided in Appendix~\ref{lemmaproof}.
This result is related to, and contains,~\cite[Lemma~1]{ChiaElGamal}. 
In particular,~\cite{ChiaElGamal} considers a single-input channel and explores the properties of codebooks driven by this input, while observing an output $Z$. In contrast, this paper's Lemma~\ref{lemma1} develops a corresponding result for a {\em multiple-access channel} with respect to $Z$,  motivated by the two-transmitters present in the model of this paper. This accounts for the new features of our Lemma~\ref{lemma1}, namely three rate constraints instead of one, as well as monitoring the entropy of two index random variables instead of one. Furthermore, the present result has one additional layer of conditioning to allow for indirect decoding of multiple confidential messages in the sequel, while in~\cite{ChiaElGamal} only one confidential message is decoded. 
\begin{remark}
In addition to establishing the main results of this paper, Lemma~\ref{lemma1} also has broader implications on the necessity of prefixing in multi-transmitter secrecy problems~\cite{ITSliang} and deriving the minimum amount of randomness needed to achieve secrecy. Csisz\'{a}r and K\"{o}rner introduced prefixing in~\cite{BCC:IT78} to expand the achievable rate region of the non-degraded broadcast channel with confidential messages, a technique that was subsequently used in essentially the same manner in multi-transmitter settings. Subsequently, Chia and El~Gamal showed that in a single-transmitter wiretap channel, prefixing can be replaced with superposition coding \cite{ChiaElGamal}. Appendix~\ref{thMAWCproof} extends this concept to a multi-transmitter setting and presents an achievability technique for the multiple access wiretap channel that utilizes minimal randomness and matches the best known achievable rates without prefixing. 
\end{remark}


\begin{theorem}
\label{thper}
An inner bound on the secrecy capacity region of the two-transmitter two-receiver channel with confidential messages is given by the set of non-negative rate pairs $(R_1,R_2)$ such that
\begin{align*}
    &R_1<\mi(U_0,U_1;Y_1|Q,V_0,V_1)-\mi(U_0;Z|Q)-\mi(U_1;Z|U_0,V_0)\\
    &R_1<\mi(U_0,U_2;Y_2|Q,V_0,V_2)-\mi(U_0;Z|Q)-\mi(U_2;Z|U_0,V_0)\\
    &R_1<\mi(U_0,U_1,V_1;Y_1|Q,V_0)-\mi(U_0;Z|Q)-\mi(U_1,V_1;Z|U_0,V_0)\\
    &R_1<\mi(U_0,U_2,V_2;Y_2|Q,V_0)-\mi(U_0;Z|Q)-\mi(U_2,V_2;Z|U_0,V_0)\\
    &R_2<\mi(V_0,V_1;Y_1|Q,U_0,U_1)-\mi(V_0;Z|Q)-\mi(V_1;Z|U_0,V_0)\\
    &R_2<\mi(V_0,V_2;Y_2|Q,U_0,U_2)-\mi(V_0;Z|Q)-\mi(V_2;Z|U_0,V_0)\\
    &R_2<\mi(U_1,V_0,V_1;Y_1|Q,U_0)-\mi(V_0;Z|Q)-\mi(U_1,V_1;Z|U_0,V_0)\\
    &R_2<\mi(U_2,V_0,V_2;Y_2|Q,U_0)-\mi(V_0;Z|Q)-\mi(U_2,V_2;Z|U_0,V_0)\\
    &R_1+R_2<\mi(U_0,U_1,V_0,V_1;Y_1|Q)-\mi(U_0,U_1,V_0,V_1;Z|Q)\\
    &R_1+R_2<\mi(U_0,U_2,V_0,V_2;Y_2|Q)-\mi(U_0,U_2,V_0,V_2;Z|Q)\\
    &R_1+R_2<\mi(U_0,U_1;Y_1|Q,V_0,V_1)+\mi(U_1,V_0,V_1;Y_1|Q,U_0)\nonumber\\
    &\quad-\mi(U_0,U_1,V_0,V_1;Z|Q)-\mi(U_1;Z|U_0,V_0)\\
    &R_1+R_2<\mi(U_0,U_1;Y_1|Q,V_0,V_1)+\mi(V_0,V_2;Y_2|Q,U_0,U_2)\nonumber\\
    &\quad-\mi(U_0,V_0;Z|Q)-\mi(U_1;Z|U_0,V_0)-\mi(V_2;Z|U_0,V_0)\\
    &R_1+R_2<\mi(U_0,U_1;Y_1|Q,V_0,V_1)+\mi(U_2,V_0,V_2;Y_2|Q,U_0)\nonumber\\
    &\quad-\mi(U_1;Z|U_0,V_0)-\mi(U_0,U_2,V_0,V_2;Z|Q)\\
    &R_1+R_2<\mi(V_0,V_1;Y_1|Q,U_0,U_1)+\mi(U_0,U_1,V_1;Y_1|Q,V_0)\nonumber\\
    &\quad-\mi(U_0,U_1,V_0,V_1;Z|Q)-\mi(V_1;Z|U_0,V_0)\\
    &R_1+R_2<\mi(V_0,V_1;Y_1|Q,U_0,U_1)+\mi(U_0,U_2;Y_2|Q,V_0,V_2)\nonumber\\
    &\quad-\mi(U_0,V_0;Z|Q)-\mi(V_1;Z|U_0,V_0)-\mi(U_2;Z|U_0,V_0)\\
    &R_1+R_2<\mi(V_0,V_1;Y_1|Q,U_0,U_1)+\mi(U_0,U_2,V_2;Y_2|Q,V_0)\nonumber\\
    &\quad-\mi(V_1;Z|U_0,V_0)-\mi(U_0,U_2,V_0,V_2;Z|Q)\\
    &R_1+R_2<\mi(U_0,U_2;Y_2|Q,V_0,V_2)+\mi(U_1,V_0,V_1;Y_1|Q,U_0)\nonumber\\
    &\quad-\mi(U_0,U_1,V_0,V_1;Z|Q)-\mi(U_2;Z|U_0,V_0)\\
    &R_1+R_2<\mi(U_0,U_2;Y_2|Q,V_0,V_2)+\mi(U_2,V_0,V_2;Y_2|Q,U_0)\nonumber\\
    &\quad-\mi(U_0,U_2,V_0,V_2;Z|Q)-\mi(U_2;Z|U_0,V_0)\\
    &R_1+R_2<\mi(V_0,V_2;Y_2|Q,U_0,U_2)+\mi(U_0,U_2,V_2;Y_2|Q,V_0)\nonumber\\
    &\quad-\mi(U_0,U_2,V_0,V_2;Z|Q)-\mi(V_2;Z|U_0,V_0)\\
    &R_1+R_2<\mi(V_0,V_2;Y_2|Q,U_0,U_2)+\mi(U_0,U_1,V_1;Y_1|Q,V_0)\nonumber\\
    &\quad-\mi(U_0,U_1,V_0,V_1;Z|Q)-\mi(V_2;Z|U_0,V_0)\\
    &R_1+R_2<\mi(U_0,U_1,V_1;Y_1|Q,V_0)+\mi(U_1,V_0,V_1;Y_1|Q,U_0)\nonumber\\
    &\quad-\mi(U_0,V_0;Z|Q)-2\mi(U_1,V_1;Z|U_0,V_0)\\
    &R_1+R_2<\mi(U_0,U_1,V_1;Y_1|Q,V_0)+\mi(U_2,V_0,V_2;Y_2|Q,U_0)\nonumber\\
    &\quad-\mi(U_0,V_0;Z|Q)-\mi(U_1,V_1;Z|U_0,V_0)-\mi(U_2,V_2;Z|U_0,V_0)\\
    &R_1+R_2<\mi(U_1,V_0,V_1;Y_1|Q,U_0)+\mi(U_0,U_2,V_2;Y_2|Q,V_0)\nonumber\\
    &\quad-\mi(U_0,U_1,V_0,V_1;Z|Q)-\mi(U_2,V_2;Z|U_0,V_0)\\
    &R_1+R_2<\mi(U_0,U_2,V_2;Y_2|Q,V_0)+\mi(U_2,V_0,V_2;Y_2|Q,U_0)\nonumber\\
    &\quad-\mi(U_0,V_0;Z|Q)-2\mi(U_2,V_2;Z|U_0,V_0)
\end{align*}
for some
\begin{align}
&p(q)p(u_0|q)p(u_1,u_2|u_0)p(v_0|q)p(v_1,v_2|v_0)\nonumber\\
&p(x_1|u_0,u_1,u_2)p(x_2|v_0,v_1,v_2)p(y_1,y_2,z|x_1,x_2),
\label{distachiperf}
\end{align}such that
\begin{align}
&\mi(U_1,U_2,V_1,V_2;Z|U_0,V_0) \leq  \mi(U_1,V_1;Z|U_0,V_0) \nonumber\\
& + \mi(U_2,V_2;Z|U_0,V_0) - \mi(U_1;U_2|U_0) - \mi(V_1;V_2|V_0).
\label{RAC}
\end{align}
\end{theorem}

\begin{figure}
\centering
\includegraphics[width=\Figwidth]{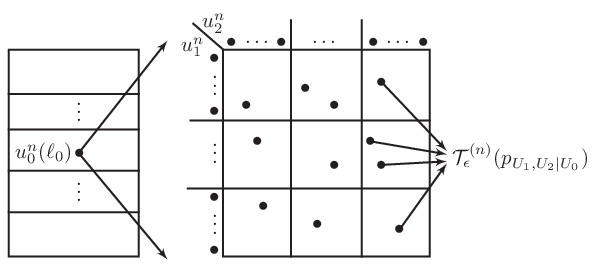}
\caption{Coding scheme for the first transmitter}
\label{CSF}
\end{figure}

The proof uses superposition coding, Wyner's wiretap coding, Marton coding, as well as indirect decoding. The details of the proof are provided in Appendix~\ref{thperproof}.

This result covers several known earlier results:
\begin{itemize}
\item
By setting $Z=\emptyset$, $U_0=U_1=U_2=X_1$, and $V_0=V_1=V_2=X_2$, the result in Theorem~\ref{thper} reduces to the capacity region of compound multiple access channel discussed in \cite{Ahlswede}.
\item
By setting $Y_2= \emptyset$ (or $Y_1= \emptyset$), $U_0=U_1=U_2=X_1$ and $V_0=V_1=V_2=X_2$, the result in Theorem~\ref{thper} reduces to the achievable rate region of multiple access wiretap channel without common message \cite{GMAWC,GMAWCJamming,YassaeeMAWC}.
\item
By setting $X_2=\emptyset$ (or $X_1=\emptyset$), $U_0=U_1=U_2$, and $Y_2=\emptyset$ (or $Y_1=\emptyset$), the result in Theorem~\ref{thper} reduces to the capacity region of broadcast channel with confidential message \cite[Corollary~2]{CsiszarKorner}.
\item
By setting $X_2=\emptyset$ (or $X_1=\emptyset$), the result in Theorem~\ref{thper} reduces to the achievable rate region for two-receiver, one-eavesdropper wiretap channel presented in \cite[Theorem~1]{ChiaElGamal}.
\end{itemize}
\begin{remark}
\label{RemRed}
By doing some algebraic manipulation we can  show that the constraint in \eqref{RAC} holds only if
\begin{align}
\mi(U_1,V_1;U_2,V_2|U_0,V_0,Z) = 0.
\label{EXCE}
\end{align}Intuitively speaking, \eqref{EXCE} shows that the Marton coding codebooks remain independent even if the eavesdropper has access to the the cloud centers.
\end{remark}


\begin{corollary}
\label{AchiCoro}
An inner bound on the secrecy capacity region of degraded two-transmitter two-receiver channel with confidential messages (Definition~\ref{DegradedDefi}) is given by the set of non-negative rate pairs $(R_1,R_2)$ such that
\begin{align}
R_1 &\le \mi(U_0;Y_2|V_0,Q) - \mi(U_0;Z|Q)\\
R_2 &\le \mi(V_0;Y_2|U_0,Q) - \mi(V_0;Z|Q)\\
R_1 + R_2 &\le \mi(U_0,V_0;Y_2|Q)-\mi(U_0;Z|Q)-\mi(V_0;Z|Q)
\end{align}
for some
\begin{align}
\label{distachiperfdeg}
p(q)p(u_0|q)p(v_0|q)p(x_1|u_0)p(x_2|v_0).
\end{align}
\end{corollary}
\begin{proof}
The proof follows from Theorem~\ref{thper} by setting $U_0=U_1=U_2$ and $V_0=V_1=V_2$ and considering the fact that the channel is degraded.
\end{proof}

\section{An Outer Bound for the Degraded Model}
\label{outerssec}
We develop an outer bound for the degraded version of the model and provide an example in which it meets the inner bound of Theorem~\ref{thper}.

 \begin{definition}
\label{DegradedDefi}
 The degraded two-transmitter two-receiver channel with confidential messages obeys:
 \begin{equation}
\label{DegradedDist}
 p(y_1,y_2,z|x_1,x_2) = p(y_1|x_1,x_2)p(y_2|y_1)p(z|y_2).
 \end{equation}
 \end{definition}
 \begin{theorem}
\label{outerdegthm}
The secrecy capacity region for the degraded two-transmitter two-receiver channel with confidential messages is included in the set of rate pairs $(R_1,R_2)$ satisfying
 \begin{align} 
 R_1 &\le \mi(U_0;Y_2|Q) - \mi(U_0;Z|Q),\label{DR1O}\\
 R_2 &\le \mi(V_0;Y_2|Q) - \mi(V_0;Z|Q),\label{DR2O}\\
 R_1 + R_2 &\le \mi(U_0,V_0;Y_2|Q) - \mi(U_0,V_0;Z|Q),\label{DR1R2O}
 \end{align}
 for some joint distribution
 \begin{align}
 p(q)p(u_0,v_0|q)p(x_1|u_0)p(x_2|v_0).
\label{outerdegdist}
 \end{align}
 \end{theorem}

\begin{proof}
\label{outerdegproof}
The proof is provided by the authors in~\cite[Section~IV]{AllertonPaper} and is omitted here for brevity.
\end{proof}

{\em {Example (Degraded Switch Model):}}
\label{firstexample}
We consider an example of the two-transmitter two-receiver channel where the first legitimate receiver has access to the noisy version of each of the two transmitted values in a time-sharing (switched) manner, without interference from the other transmitter (Fig.~\ref{fig2}). The second legitimate receiver has access to a noisy version of the first receiver, and the eavesdropper has access to a noisy version of the second receiver. 
The switch channel state information is made available to all terminals. In this model the channel outputs are as follows:
 \begin{align}
 y'_1&=(y_1,s),
\label{NewOutp1}\\
 y'_2&=(y_2,s),
\label{NewOutp2}\\
 z'&=(z,s).
\label{NewOutp3}
 \end{align}
This model consists of a channel with states that are causally available at both the encoders and decoders.

The statistics of the channel, conditioned on the switch state, are expressed as follows:
\begin{align}
p(&y'_1,y'_2,z|x_1,,x_2,s) \twocolbreak
 =p(y_1|x_1,x_2,s)\, p(y_2|y_1,s)\, p(z|y_2,s)
\end{align}


The switch model describes, e.g., frequency hopping over two frequencies~\cite{LiuMaric}. 
The state (switch) is a binary random variable that chooses between listening to the Transmitter~1, with probability $\tau$, and listening to the Transmitter~2, with probability $1-\tau$, independently at each time slot. We further assume the state is i.i.d.\ across time.
 \begin{figure}
 \centering
 \includegraphics[width=\Figwidth]{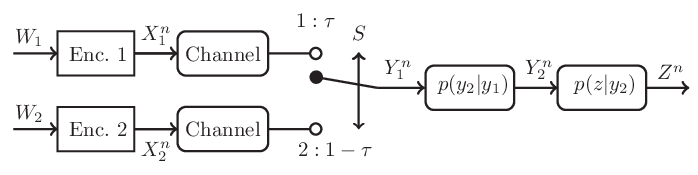}
 \caption{Degraded switch model}
 
\label{fig2}
 \end{figure}
 \begin{align}
 p(y_1|x_1,x_2,s)={} &p(y_1|x_1){\indic{1}}\{s=1\} + p(y_1|x_2){\indic{1}}\{s=2\}\nonumber
 \end{align}
where $\indic{1}$ is the indicator function.
 \begin{theorem}
\label{outerdegswithm}
The secrecy capacity region for the degraded switch two-transmitter two-receiver channel with confidential messages, is given by the set of rate pairs $(R_1,R_2)$ satisfying
 \begin{align}
 R_1 &\le \mi(U_0;Y'_2|V_0,Q) - \mi(U_0;Z'|Q),\label{R1O}\\
 R_2 &\le \mi(V_0;Y'_2|U_0,Q) - \mi(V_0;Z'|Q),\label{R2O}\\
 R_1 + R_2 &\le \mi(U_0,V_0;Y'_2|Q) - \mi(U_0,V_0;Z'|Q),\label{R1R2O}
 \end{align}
 for some joint distribution
 \begin{align}
 p(q)p(u_0|q)p(v_0|q)p(x_1|u_0)p(x_2|v_0).
\label{outerdegswidist}
 \end{align}
 \end{theorem}

\begin{proof}
\label{outerproof}
The proof is available in \cite[Section~IV]{AllertonPaper}.
\end{proof}

 \section{A General Outer Bound}
\label{Gouterssec}
We now develop a general outer bound for the model of Fig.~\ref{fig} and provide an example in which it meets the inner bound  of Theorem~\ref{thper}.

 \begin{theorem}
\label{Gouterthm}
The secrecy capacity region for the two-transmitter two-receiver channel with confidential messages is included in the set of rate pairs $(R_1,R_2)$ satisfying
\begin{align}
R_1 &\le \mi(U_0;Y_1,Y_2|Q) - \mi(U_0;Z|Q),\label{GR1O}\\
R_2 &\le \mi(V_0;Y_1,Y_2|Q) - \mi(V_0;Z|Q),\label{GR2O}\\
R_1 + R_2 &\le \mi(U_0,V_0;Y_1,Y_2|Q) - \mi(U_0,V_0;Z|Q),\label{GR1R2O}
\end{align}
for some joint distribution
\begin{align}
p(q)p(u_0,v_0|q)p(x_1|u_0)p(x_2|v_0).
\label{Gouterdist}
\end{align}
\end{theorem}
\begin{proof}
\label{Gouterproof}
The proof is available in \cite[Section~V]{AllertonPaper}.
\end{proof}

 {\em {Example (Noiseless Switch Model):}}
\label{secondexample}
\begin{figure}
\centering
\includegraphics[width=8.5cm]{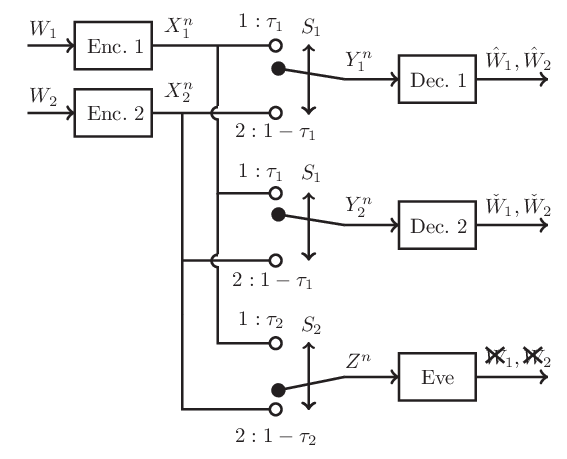}
\caption{Noiseless switch model}

\label{SKD3}
\end{figure}
This example is motivated by two transmitters operating on different spectral bands, while the receiving terminals may receive adaptively on one band at a time~\cite{LiuMaric}. The eavesdropper in our example has access to one noiseless interference-free transmitted value at a time. Here, it is assumed that both legitimate receivers operate according to a common random switch $s_1$ that is connected to Transmitter~1 with probability $\tau_1$ and to Transmitter~2 with probability $1-\tau_1$, and the eavesdropper operates according to another random switch $s_2$ that is connected to Transmitter~1 with probability $\tau_2$ and to Transmitter~2 with probability $1-\tau_2$. Aside from the switches, the channel is noiseless. Both receivers and the eavesdropper have access to their own switch state information.
Therefore the channel outputs are considered
\begin{align}
y'_1&=(y_1,s_1),
\label{NewOutp12}\\
y'_2&=(y_2,s_1),
\label{NewOutp22}\\
z'&=(z,s_2).
\label{NewOutp32}
\end{align}
Since $y_1=y_2$, we also have $y'_1=y'_2$.
\begin{theorem}
\label{outergeneralswithm}
The secrecy capacity region for the noiseless switch two-transmitter two-receiver channel with confidential messages is given by the set of rate pairs $(R_1,R_2)$ satisfying
\begin{align}
R_1 &\le (\tau_1-\tau_2)^{+}\ent(X_1),\label{R1OGSC}\\
R_2 &\le (\tau_2-\tau_1)^{+}\ent(X_2),\label{R2OGSC}
\end{align}
\end{theorem}
where $(x)^+=\max\{0,x\}$.
\begin{proof}
\label{SCouterproof}
The proof is available in \cite[Section~V]{AllertonPaper}.
\end{proof}
The capacity region in Theorem~\ref{outergeneralswithm} shows that transmitters can securely communicate to receivers as long as $\tau_1\ne\tau_2$.

\section{Strong Secrecy}
\label{AchivRateRegionStrong}
\begin{theorem}
\label{thstr}
An inner bound on the secrecy capacity region of the two-transmitter two-receiver channel with confidential messages is given by the set of non-negative rate pairs $(R_1,R_2)$ such that
\begin{align*}
&R_1<\mi(U_0,U_1;Y_1|Q,V_0,V_1)-\mi(U_0;Z|Q)-\mi(U_1;Z|U_0,V_0)\\
&R_1<\mi(U_0,U_2;Y_2|Q,V_0,V_2)-\mi(U_0;Z|Q)-\mi(U_2;Z|U_0,V_0)\\
&R_1<\mi(U_0,U_1,V_1;Y_1|Q,V_0)-\mi(U_0;Z|Q)-\mi(U_1,V_1;Z|U_0,V_0)\\
&R_1<\mi(U_0,U_2,V_2;Y_2|Q,V_0)-\mi(U_0;Z|Q)-\mi(U_2,V_2;Z|U_0,V_0)\\
&2R_1<\mi(U_0,U_1;Y_1|Q,V_0,V_1)+\mi(U_0,U_2;Y_2|Q,V_0,V_2)\\
&\quad-2\mi(U_0;Z|Q)-\mi(U_1,U_2;Z|U_0,V_0)-\mi(U_1;U_2|U_0)\\
&2R_1<\mi(U_0,U_1;Y_1|Q,V_0,V_1)+\mi(U_0,U_2,V_2;Y_2|Q,V_0)\\
&\quad-2\mi(U_0;Z|Q)-\mi(U_1,U_2,V_2;Z|U_0,V_0)-\mi(U_1;U_2|U_0)\\
&2R_1<\mi(U_0,U_2;Y_2|Q,V_0,V_2)+\mi(U_0,U_1,V_1;Y_1|Q,V_0)\\
&\quad-2\mi(U_0;Z|Q)-\mi(U_1,U_2,V_1;Z|U_0,V_0)-\mi(U_1;U_2|U_0)\\
&2R_1<\mi(U_0,U_1,V_1;Y_1|Q,U_0)+\mi(U_0,U_2,V_2;Y_2|Q,U_0)\\
&\quad-2\mi(U_0;Z|Q)-\mi(U_1,U_2,V_1,V_2;Z|U_0,V_0)-\mi(U_1,V_1;U_2,V_2|U_0,V_0)\\
&R_2<\mi(V_0,V_1;Y_1|Q,U_0,U_1)-\mi(V_0;Z|Q)-\mi(V_1;Z|U_0,V_0)\\
&R_2<\mi(V_0,V_2;Y_2|Q,U_0,U_2)-\mi(V_0;Z|Q)-\mi(V_2;Z|U_0,V_0)\\
&R_2<\mi(U_1,V_0,V_1;Y_1|Q,U_0)-\mi(V_0;Z|Q)-\mi(U_1,V_1;Z|U_0,V_0)\\
&R_2<\mi(U_2,V_0,V_2;Y_2|Q,U_0)-\mi(V_0;Z|Q)-\mi(U_2,V_2;Z|U_0,V_0)\\
&2R_2<\mi(V_0,V_1;Y_1|Q,U_0,U_1)+\mi(V_0,V_2;Y_2|Q,U_0,U_2)\\
&\quad-2\mi(V_0;Z|Q)-\mi(V_1,V_2;Z|U_0,V_0)-\mi(V_1;V_2|V_0)\\
&2R_2<\mi(V_0,V_1;Y_1|Q,U_0,U_1)+\mi(U_2,V_0,V_2;Y_2|Q,U_0)\\
&\quad-2\mi(V_0;Z|Q)-\mi(U_2,V_1,V_2;Z|U_0,V_0)-\mi(V_1;V_2|V_0)\\
&2R_2<\mi(V_0,V_2;Y_2|Q,U_0,U_2)+\mi(U_1,V_0,V_1;Y_1|Q,U_0)\\
&\quad-2\mi(V_0;Z|Q)-\mi(U_1,V_1,V_2;Z|U_0,V_0)-\mi(V_1;V_2|V_0)\\
&2R_2<\mi(U_1,V_0,V_1;Y_1|Q,U_0)+\mi(U_2,V_0,V_2;Y_2|Q,U_0)\\
&\quad-2\mi(V_0;Z|Q)-\mi(U_1,U_2,V_1,V_2;Z|U_0,V_0)-\mi(U_1,V_1;U_2,V_2|U_0,V_0)\\
&R_1+R_2<\mi(U_0,U_1;Y_1|Q,V_0,V_1)+\mi(V_0,V_2;Y_2|Q,U_0,U_2)\\
&\quad-\mi(U_0,V_0;Z|Q)-\mi(U_1,V_2;Z|U_0,V_0)\\
&R_1+R_2<\mi(U_0,U_1;Y_1|Q,V_0,V_1)+\mi(U_1,V_0,V_1;Y_1|Q,U_0)\\
&\quad-\mi(U_0,V_0;Z|Q)-\mi(U_1;Z|U_0,V_0)-\mi(U_1,V_1;Z|U_0,V_0)\\
&R_1+R_2<\mi(U_0,U_1;Y_1|Q,V_0,V_1)+\mi(U_2,V_0,V_2;Y_2|Q,U_0)\\
&\quad-\mi(U_0,V_0;Z|Q)-\mi(U_1,U_2,V_2;Z|U_0,V_0)-\mi(U_1;U_2|U_0)\\
&R_1+R_2<\mi(U_0,U_2;Y_2|Q,V_0,V_2)+\mi(V_0,V_1;Y_1|Q,U_0,U_1)\\
&\quad-\mi(U_0,V_0;Z|Q)-\mi(U_2,V_1;Z|U_0,V_0)\\
&R_1+R_2<\mi(U_0,U_2;Y_2|Q,V_0,V_2)+\mi(U_1,V_0,V_1;Y_1|Q,U_0)\\
&\quad-\mi(U_0,V_0;Z|Q)-\mi(U_1,U_2,V_1;Z|U_0,V_0)-\mi(U_1;U_2|U_0)\\
&R_1+R_2<\mi(U_0,U_2;Y_2|Q,V_0,V_2)+\mi(U_2,V_0,V_2;Y_2|Q,U_0)\\
&\quad-\mi(U_0,V_0;Z|Q)-\mi(U_2;Z|U_0,V_0)-\mi(U_2,V_2;Z|U_0,V_0)\\
&R_1+R_2<\mi(V_0,V_1;Y_1|Q,U_0,U_1)+\mi(U_0,U_1,V_1;Y_1|Q,V_0)\\
&\quad-\mi(U_0,V_0;Z|Q)-\mi(V_1;Z|U_0,V_0)-\mi(U_1,V_1;Z|U_0,V_0)\\
&R_1+R_2<\mi(V_0,V_1;Y_1|Q,U_0,U_1)+\mi(U_0,U_2,V_2;Y_2|Q,V_0)\\
&\quad-\mi(U_0,V_0;Z|Q)-\mi(U_2,V_1,V_2;Z|U_0,V_0)-\mi(V_1;V_2|V_0)\\
&R_1+R_2<\mi(V_0,V_2;Y_2|Q,U_0,U_2)+\mi(U_0,U_1,V_1;Y_1|Q,V_0)\\
&\quad-\mi(U_0,V_0;Z|Q)-\mi(U_1,V_1,V_2;Z|U_0,V_0)-\mi(V_1;V_2|V_0)\\
&R_1+R_2<\mi(V_0,V_2;Y_2|Q,U_0,U_2)+\mi(U_0,U_2,V_2;Y_2|Q,V_0)\\
&\quad-\mi(U_0,V_0;Z|Q)-\mi(V_2;Z|U_0,V_0)-\mi(U_2,V_2;Z|U_0,V_0)\\
&R_1+R_2<\mi(U_0,U_1,V_1;Y_1|Q,V_0)+\mi(U_1,V_0,V_1;Y_1|Q,U_0)\\
&\quad-\mi(U_0,V_0;Z|Q)-2\mi(U_1,V_1;Z|U_0,V_0)\\
&R_1+R_2<\mi(U_0,U_1,V_1;Y_1|Q,V_0)+\mi(U_2,V_0,V_2;Y_2|Q,U_0)-\mi(U_0,V_0;Z|Q)\\
&\quad-\mi(U_1,U_2,V_1,V_2;Z|U_0,V_0)-\mi(U_1,V_1;U_2,V_2|U_0,V_0)\\
&R_1+R_2<\mi(U_0,U_2,V_2;Y_2|Q,V_0)+\mi(U_1,V_0,V_1;Y_1|Q,U_0)-\mi(U_0,V_0;Z|Q)\\
&\quad-\mi(U_1,U_2,V_1,V_2;Z|U_0,V_0)-\mi(U_1,V_1;U_2,V_2|U_0,V_0)\\
&R_1+R_2<\mi(U_0,U_2,V_2;Y_2|Q,V_0)+\mi(U_2,V_0,V_2;Y_2|Q,U_0)\\
&\quad-\mi(U_0,V_0;Z|Q)-2\mi(U_2,V_2;Z|U_0,V_0)\\
&R_1+R_2<\mi(U_0,U_2,V_0,V_2;Y_2|Q)-\mi(U_0,U_2,V_0,V_2;Z|Q)\\
&R_1+R_2<\mi(U_0,U_1,V_0,V_1;Y_1|Q)-\mi(U_0,U_1,V_0,V_1;Z|Q)\\
&2R_1+R_2<\mi(U_0,U_1;Y_1|Q,V_0,V_1)+\mi(U_0,U_2,V_0,V_2;Y_2|Q)\\
&\quad-\mi(U_0;Z|Q)-\mi(U_0,V_0;Z|Q)-\mi(U_1,U_2,V_2;Z|U_0,V_0)-\mi(U_1;U_2|U_0)\\
&2R_1+R_2<\mi(U_0,U_2;Y_2|Q,V_0,V_2)+\mi(U_0,U_1,V_0,V_1;Y_1|Q)\\
&\quad-\mi(U_0;Z|Q)-\mi(U_0,V_0;Z|Q)-\mi(U_1,U_2,V_1;Z|U_0,V_0)-\mi(U_1;U_2|U_0)\\
&2R_1+R_2<\mi(U_0,U_1,V_1;Y_1|Q,V_0)+\mi(U_0,U_2,V_0,V_2;Y_2|Q)\\
&\quad-\mi(U_0;Z|Q)-\mi(U_0,V_0;Z|Q)-\mi(U_1,U_2,V_1,V_2;Z|U_0,V_0)-\mi(U_1,V_1;U_2,V_2|U_0,V_0)\\
&2R_1+R_2<\mi(U_0,U_2,V_2;Y_2|Q,V_0)+\mi(U_0,U_1,V_0,V_1;Y_1|Q)\\
&\quad-\mi(U_0;Z|Q)-\mi(U_0,V_0;Z|Q)-\mi(U_1,U_2,V_1,V_2;Z|U_0,V_0)-\mi(U_1,V_1;U_2,V_2|U_0,V_0)\\
&R_1+2R_2<\mi(V_0,V_1;Y_1|Q,U_0,U_1)+\mi(U_0,U_2,V_0,V_2;Y_2|Q)\\
&\quad-\mi(V_0;Z|Q)-\mi(U_0,V_0;Z|Q)-\mi(U_2,V_1,V_2;Z|U_0,V_0)-\mi(V_1;V_2|V_0)\\
&R_1+2R_2<\mi(V_0,V_2;Y_2|Q,U_0,U_2)+\mi(U_0,U_1,V_0,V_1;Y_1|Q)\\
&\quad-\mi(V_0;Z|Q)-\mi(U_0,V_0;Z|Q)-\mi(U_1,V_1,V_2;Z|U_0,V_0)-\mi(V_1;V_2|V_0)\\
&R_1+2R_2<\mi(U_0,U_1,V_0,V_1;Y_1|Q)+\mi(U_2,V_0,V_2;Y_2|Q,U_0)\\
&\quad-\mi(V_0;Z|Q)-\mi(U_0,V_0;Z|Q)-\mi(U_1,U_2,V_1,V_2;Z|U_0,V_0)-\mi(U_1,V_1;U_2,V_2|U_0,V_0)\\
&R_1+2R_2<\mi(U_1,V_0,V_1;Y_1|Q,U_0)+\mi(U_0,U_2,V_0,V_2;Y_2|Q)\\
&\quad-\mi(V_0;Z|Q)-\mi(U_0,V_0;Z|Q)-\mi(U_1,U_2,V_1,V_2;Z|U_0,V_0)-\mi(U_1,V_1;U_2,V_2|U_0,V_0)\\
&2R_1+2R_2<2\mi(U_0,U_1;Y_1|Q,V_0,V_1)+\mi(V_0,V_2;Y_2|Q,U_0,U_2)+\mi(U_1,V_0,V_1;Y_1|Q,U_0)\\
&\quad-2\mi(U_0,V_0;Z|Q)-2\mi(U_1;Z|U_0,V_0)-\mi(U_1,V_1,V_2;Z|U_0,V_0)-\mi(V_1;V_2|V_0)\\
&2R_1+2R_2<2\mi(U_0,U_1;Y_1|Q,V_0,V_1)+\mi(U_1,V_0,V_1;Y_1|Q,U_0)+\mi(U_2,V_0,V_2;Y_2|Q,U_0)\\
&\quad-2\mi(U_0,V_0;Z|Q)-2\mi(U_1;Z|U_0,V_0)-\mi(U_1,U_2,V_1,V_2;Z|U_0,V_0)-\mi(U_1,V_1;U_2,V_2|U_0,V_0)\\
&2R_1+2R_2<\mi(U_0,U_1;Y_1|Q,V_0,V_1)+\mi(U_0,U_2;Y_2|Q,V_0,V_2)\\
&\quad+\mi(U_1,V_0,V_1;Y_1|Q,U_0)+\mi(U_2,V_0,V_2;Y_2|Q,U_0)-2\mi(U_0,V_0;Z|Q)-\mi(U_1,U_2;Z|U_0,V_0)\\
&\quad-\mi(U_1,U_2,V_1,V_2;Z|U_0,V_0)-\mi(U_1;U_2|U_0)-\mi(U_1,V_1;U_2,V_2|U_0,V_0)\\
&2R_1+2R_2<\mi(U_0,U_1;Y_1|Q,V_0,V_1)+2\mi(V_0,V_2;Y_2|Q,U_0,U_2)+\mi(U_0,U_2,V_2;Y_2|Q,V_0)\\
&\quad-2\mi(U_0,V_0;Z|Q)-2\mi(V_2;Z|U_0,V_0)-\mi(U_1,U_2,V_2;Z|U_0,V_0)-\mi(U_1;U_2|U_0)\\
&2R_1+2R_2<\mi(U_0,U_1;Y_1|Q,V_0,V_1)+\mi(V_0,V_2;Y_2|Q,U_0,U_2)+\mi(U_0,U_1,V_0,V_1;Y_1|Q)\\
&\quad-2\mi(U_0,V_0;Z|Q)-\mi(U_1;Z|U_0,V_0)-\mi(U_1,V_1,V_2;Z|U_0,V_0)-\mi(V_1;V_2|V_0)\\
&2R_1+2R_2<\mi(U_0,U_1;Y_1|Q,V_0,V_1)+\mi(V_0,V_2;Y_2|Q,U_0,U_2)+\mi(U_0,U_2,V_0,V_2;Y_2|Q)\\
&\quad-2\mi(U_0,V_0;Z|Q)-\mi(V_2;Z|U_0,V_0)-\mi(U_1,U_2,V_2;Z|U_0,V_0)-\mi(U_1;U_2|U_0)\\
&2R_1+2R_2<\mi(U_0,U_1;Y_1|Q,V_0,V_1)+\mi(U_1,V_0,V_1;Y_1|Q,U_0)+\mi(U_0,U_2,V_0,V_2;Y_2|Q)\\
&\quad-2\mi(U_0,V_0;Z|Q)-\mi(U_1;Z|U_0,V_0)-\mi(U_1,U_2,V_1,V_2;Z|U_0,V_0)-\mi(U_1,V_1;U_2,V_2|U_0,V_0)\\
&2R_1+2R_2<\mi(U_0,U_1;Y_1|Q,V_0,V_1)+\mi(U_2,V_0,V_2;Y_2|Q,U_0)+\mi(U_0,U_1,V_0,V_1;Y_1|Q)\\
&\quad-2\mi(U_0,V_0;Z|Q)-\mi(U_1;Z|U_0,V_0)-\mi(U_1,U_2,V_1,V_2;Z|U_0,V_0)-\mi(U_1,V_1;U_2,V_2|U_0,V_0)\\
&2R_1+2R_2<\mi(U_0,U_2;Y_2|Q,V_0,V_2)+2\mi(V_0,V_1;Y_1|Q,U_0,U_1)+\mi(U_0,U_1,V_1;Y_1|Q,V_0)\\
&\quad-2\mi(U_0,V_0;Z|Q)-2\mi(V_1;Z|U_0,V_0)-\mi(U_1,U_2,V_1;Z|U_0,V_0)-\mi(U_1;U_2|U_0)\\
&2R_1+2R_2<2\mi(U_0,U_2;Y_2|Q,V_0,V_2)+\mi(V_0,V_1;Y_1|Q,U_0,U_1)+\mi(U_2,V_0,V_2;Y_2|Q,U_0)\\
&\quad-2\mi(U_0,V_0;Z|Q)-2\mi(U_2;Z|U_0,V_0)-\mi(U_2,V_1,V_2:Z|U_0,V_0)-\mi(V_1;V_2|V_0)\\
&2R_1+2R_2<2\mi(U_0,U_2;Y_2|Q,V_0,V_2)+\mi(U_1,V_0,V_1;Y_1|Q,U_0)+\mi(U_2,V_0,V_2;Y_2|Q,U_0)\\
&\quad-2\mi(U_0,V_0;Z|Q)-2\mi(U_2;Z|U_0,V_0)-\mi(U_1,U_2,V_1,V_2;Z|U_0,V_0)-\mi(U_1,V_1;U_2,V_2|U_0,V_0)\\
&2R_1+2R_2<\mi(U_0,U_2;Y_2|Q,V_0,V_2)+\mi(V_0,V_1;Y_1|Q,U_0,U_1)+\mi(U_0,U_1,V_0,V_1;Y_1|Q)\\
&\quad-2\mi(U_0,V_0;Z|Q)-\mi(V_1;Z|U_0,V_0)-\mi(U_1,U_2,V_1;Z|U_0,V_0)-\mi(U_1;U_2|U_0)\\
&2R_1+2R_2<\mi(U_0,U_2;Y_2|Q,V_0,V_2)+\mi(V_0,V_1;Y_1|Q,U_0,U_1)+\mi(U_0,U_2,V_0,V_2;Y_2|Q)\\
&\quad-2\mi(U_0,V_0;Z|Q)-\mi(U_2;Z|U_0,V_0)-\mi(U_2,V_1,V_2;Z|U_0,V_0)-\mi(V_1;V_2|V_0)\\
&2R_1+2R_2<\mi(U_0,U_2;Y_2|Q,V_0,V_2)+\mi(U_1,V_0,V_1;Y_1|Q,U_0)+\mi(U_0,U_2,V_0,V_2;Y_2|Q)\\
&\quad-2\mi(U_0,V_0;Z|Q)-\mi(U_2;Z|U_0,V_0)-\mi(U_1,U_2,V_1,V_2;Z|U_0,V_0)-\mi(U_1,V_1;U_2,V_2|U_0,V_0)\\
&2R_1+2R_2<\mi(U_0,U_2;Y_2|Q,V_0,V_2)+\mi(U_2,V_0,V_2;Y_2|Q,U_0)+\mi(U_0,U_1,V_0,V_1;Y_1|Q)\\
&\quad-2\mi(U_0,V_0;Z|Q)-\mi(U_2;Z|U_0,V_0)-\mi(U_1,U_2,V_1,V_2;Z|U_0,V_0)-\mi(U_1,V_1;U_2,V_2|U_0,V_0)\\
&2R_1+2R_2<2\mi(V_0,V_1;Y_1|Q,U_0,U_1)+\mi(U_0,U_1,V_1;Y_1|Q,V_0)+\mi(U_0,U_2,V_2;Y_2|Q,V_0)\\
&\quad-2\mi(U_0,V_0;Z|Q)-2\mi(V_1;Z|U_0,V_0)-\mi(U_1,U_2,V_1,V_2;Z|U_0,V_0)-\mi(U_1,V_1;U_2,V_2|U_0,V_0)\\
&2R_1+2R_2<\mi(V_0,V_1;Y_1|Q,U_0,U_1)+\mi(V_0,V_2;Y_2|Q,U_0,U_2)\\
&\quad+\mi(U_0,U_1,V_1;Y_1|Q,V_0)+\mi(U_0,U_2,V_2;Y_2|Q,V_0)-2\mi(U_0,V_0;Z|Q)\\
&\quad-\mi(V_1,V_2;Z|U_0,V_0)-\mi(U_1,U_2,V_1,V_2;Z|U_0,V_0)-\mi(V_1;V_2|V_0)-\mi(U_1,V_1;U_2,V_2|U_0,V_0)\\
&2R_1+2R_2<\mi(V_0,V_1;Y_1|Q,U_0,U_1)+\mi(U_0,U_1,V_1;Y_1|Q,V_0)+\mi(U_0,U_2,V_0,V_2;Y_2|Q)\\
&\quad-2\mi(U_0,V_0;Z|Q)-\mi(V_1;Z|U_0,V_0)-\mi(U_1,U_2,V_1,V_2;Z|U_0,V_0)-\mi(U_1,V_1;U_2,V_2|U_0,V_0)\\
&2R_1+2R_2<\mi(V_0,V_1;Y_1|Q,U_0,U_1)+\mi(U_0,U_2,V_2;Y_2|Q,V_0)+\mi(U_0,U_1,V_0,V_1;Y_1|Q)\\
&\quad-2\mi(U_0,V_0;Z|Q)-\mi(V_1;Z|U_0,V_0)-\mi(U_1,U_2,V_1,V_2;Z|U_0,V_0)-\mi(U_1,V_1;U_2,V_2|U_0,V_0)\\
&2R_1+2R_2<2\mi(V_0,V_2;Y_2|Q,U_0,U_2)+\mi(U_0,U_1,V_1;Y_1|Q,V_0)+\mi(U_0,U_2,V_2;Y_2|Q,V_0)\\
&\quad-2\mi(U_0,V_0;Z|Q)-2\mi(V_2;Z|U_0,V_0)-\mi(U_1,U_2,V_1,V_2;Z|U_0,V_0)-\mi(U_1,V_1;U_2,V_2|U_0,V_0)\\
&2R_1+2R_2<\mi(V_0,V_2;Y_2|Q,U_0,U_2)+\mi(U_0,U_1,V_1;Y_1|Q,V_0)+\mi(U_0,U_2,V_0,V_2;Y_2|Q)\\
&\quad-2\mi(U_0,V_0;Z|Q)-\mi(V_2;Z|U_0,V_0)-\mi(U_1,U_2,V_1,V_2;Z|U_0,V_0)-\mi(U_1,V_1;U_2,V_2|U_0,V_0)\\
&2R_1+2R_2<\mi(V_0,V_2;Y_2|Q,U_0,U_2)+\mi(U_0,U_2,V_2;Y_2|Q,V_0)+\mi(U_0,U_1,V_0,V_1;Y_1|Q)\\
&\quad-2\mi(U_0,V_0;Z|Q)-\mi(V_2;Z|U_0,V_0)-\mi(U_1,U_2,V_1,V_2;Z|U_0,V_0)-\mi(U_1,V_1;U_2,V_2|U_0,V_0)\\
&2R_1+2R_2<\mi(U_0,U_1,V_0,V_1;Y_1|Q)+\mi(U_0,U_2,V_0,V_2;Y_2|Q)\\
&\quad-2\mi(U_0,V_0;Z|Q)-\mi(U_1,U_2,V_1,V_2;Z|U_0,V_0)-\mi(U_1,V_1;U_2,V_2|U_0,V_0)\\
&3R_1+3R_2<\mi(U_0,U_1;Y_1|Q,V_0,V_1)+\mi(U_1,V_0,V_1;Y_1|Q,U_0)\\
&\quad+2\mi(U_0,U_2,V_0,V_2;Y_2|Q)-3\mi(U_0,V_0;Z)-\mi(U_1,U_2,V_2;Z|U_0,V_0)\\
&\quad-\mi(U_1,U_2,V_1,V_2;Z|U_0,V_0)-\mi(U_1;U_2|U_0)-\mi(U_1,V_1;U_2,V_2|U_0,V_0)\\
&3R_1+3R_2<\mi(U_0,U_2;Y_2|Q,V_0,V_2)+\mi(U_2,V_0,V_2;Y_2|Q,U_0)\\
&\quad+2\mi(U_0,U_1,V_0,V_1;Y_1|Q)-3\mi(U_0,V_0;Z|Q)-\mi(U_1,U_2,V_1;Z|U_0,V_0)\\
&\quad-\mi(U_1,U_2,V_1,V_2;Z|U_0,V_0)-\mi(U_1;U_2|U_0)-\mi(U_1,V_1;U_2,V_2|U_0,V_0)\\
&3R_1+3R_2<\mi(V_0,V_1;Y_1|Q,U_0,U_1)+\mi(U_0,U_1,V_1;Y_1|Q,V_0)\\
&\quad+2\mi(U_0,U_2,V_0,V_2;Y_2|Q)-3\mi(U_0,V_0;Z|Q)-\mi(U_2,V_1,V_2;Z|U_0,V_0)\\
&\quad-\mi(U_1,U_2,V_1,V_2;Z|U_0,V_0)-\mi(V_1;V_2|V_0)-\mi(U_1,V_1;U_2,V_2|U_0,V_0)\\
&3R_1+3R_2<\mi(V_0,V_2;Y_2|Q,U_0,U_2)+\mi(U_0,U_2,V_2;Y_2|Q,V_0)\\
&\quad+2\mi(U_0,U_1,V_0,V_1;Y_1|Q)-3\mi(U_0,V_0;Z|Q)-\mi(U_1,V_1,V_2;Z|U_0,V_0)\\
&\quad-\mi(U_1,U_2,V_1,V_2;Z|U_0,V_0)-\mi(V_1;V_2|V_0)-\mi(U_1,V_1;U_2,V_2|U_0,V_0)\\
&3R_1+3R_2<\mi(U_0,U_1,V_1;Y_1|Q,V_0)+\mi(U_1,V_0,V_1;Y_1|Q,U_0)+2\mi(U_0,U_2,V_0,V_2;Y_2|Q)\\
&\quad-3\mi(U_0,V_0;Z|Q)-2\mi(U_1,U_2,V_1,V_2;Z|U_0,V_0)-2\mi(U_1,V_1;U_2,V_2|U_0,V_0)\\
&3R_1+3R_2<\mi(U_0,U_2,V_2;Y_2|Q,V_0)+\mi(U_2,V_0,V_2;Y_2|Q,U_0)+2\mi(U_0,U_1,V_0,V_1;Y_1|Q)\\
&\quad-3\mi(U_0,V_0;Z|Q)-2\mi(U_1,U_2,V_1,V_2;Z|U_0,V_0)-2\mi(U_1,V_1;U_2,V_2|U_0,V_0)\\
&4R_1+4R_2<\mi(U_0,U_1;Y_1|V_0,V_1)+\mi(U_0,U_2;Y_2|Q,V_0,V_2)+2\mi(U_1,V_0,V_1;Y_1|Q,U_0)\\
&\quad+2\mi(U_0,U_2,V_0,V_2;Y_2|Q)-4\mi(U_0,V_0;Z|Q)-\mi(U_1,U_2;Z|U_0,V_0)\\
&\quad-2\mi(U_1,U_2,V_1,V_2;Z|U_0,V_0)-\mi(U_1;U_2|U_0)-2\mi(U_1,V_1;U_2,V_2|U_0,V_0)\\
&4R_1+4R_2<\mi(U_0,U_1;Y_1|Q,V_0,V_1)+\mi(U_0,U_2;Y_2|Q,V_0,V_2)\\
&\quad+2\mi(U_2,V_0,V_2;Y_2|Q,U_0)+2\mi(U_0,U_1,V_0,V_1;Y_1|Q)-4\mi(U_0,V_0;Z|Q)\\
&\quad-\mi(U_1,U_2;Z|U_0,V_0)-2\mi(U_1,U_2,V_1,V_2;Z|U_0,V_0)-\mi(U_1;U_2|U_0)-2\mi(U_1,V_1;U_2,V_2|U_0,V_0)\\
&4R_1+4R_2<\mi(V_0,V_1;Y_1|Q,U_0,U_1)+\mi(V_0,V_2;Y_2|Q,U_0,U_2)\\
&\quad+2\mi(U_0,U_1,V_1;Y_1|Q,V_0)+2\mi(U_0,U_2,V_0,V_2;Y_2|Q)-4\mi(U_0,V_0;Z|Q)\\
&\quad-\mi(V_1,V_2;Z|U_0,V_0)-2\mi(U_1,U_2,V_1,V_2;Z|U_0,V_0)-\mi(V_1;V_2|V_0)-2\mi(U_1,V_1;U_2,V_2|U_0,V_0)\\
&4R_1+4R_2<\mi(V_0,V_1;Y_1|Q,U_0,U_1)+\mi(V_0,V_2;Y_2|Q,U_0,U_2)\\
&\quad+2\mi(U_0,U_2,V_2;Y_2|Q,V_0)+2\mi(U_0,U_1,V_0,V_1;Y_1|Q)-4\mi(U_0,V_0;Z|Q)\\
&\quad-\mi(V_1,V_2;Z|U_0,V_0)-2\mi(U_1,U_2,V_1,V_2;Z|U_0,V_0)-\mi(V_1;V_2|V_0)-2\mi(U_1,V_1;U_2,V_2|U_0,V_0)
\end{align*}
for some distribution
\begin{align}
&p(q)p(u_0,u_1,u_2|q)p(v_0,v_1,v_2|q)\twocolbreaktimes p(x_1|u_0,u_1,u_2)p(x_2|v_0,v_1,v_2)p(y_1,y_2,z|x_1,x_2),
\label{distachistr}
\end{align}
\end{theorem}


\begin{proof}

The achievability proof is inspired by \cite{OSRB}, and is based on solving a dual secret key agreement problem in the source model that includes shared randomness at all terminals (see Fig.~\ref{SKD}). 
In this dual model, rate constraints are derived  so that the input and output distributions of the dual model approximate that of the original model while satisfying reliability and secrecy conditions in the dual model. The probability approximation then guarantees that reliability {\em and} secrecy conditions can be achieved in the original model. Finally, it is shown that there exists one realization of shared randomness for which the above mentioned are valid, thus removing the necessity for common randomness.


We begin by developing the encoding and decoding strategies for the source model and the original model, and derive and compare the joint probability distributions arising from these two strategies.

\begin{figure}
\centering
\includegraphics[width=\Figwidth]{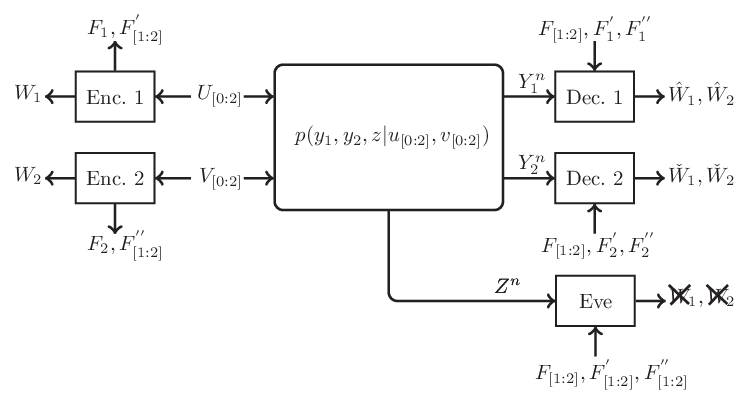}
\caption{Dual secret key agreement problem in the source model for the original problem.}

\label{SKD}
\end{figure}

We begin with the multi-terminal secret key agreement problem in the source model  as depicted in Fig.~\ref{SKD}. Let $(U_{[0:2]}^n,V_{[0:2]}^n,X_1^n,X_2^n,Y_1^n,Y_2^n,Z^n)$ be i.i.d. and distributed according to 
\begin{equation}
\label{GDistStrong}
p(u_{[0:2]},x_1)p(v_{[0:2]},x_2)p(y_1,y_2,z|x_1,x_2).
\end{equation}

{\em Random Binning:} 
\begin{itemize}
\item To {\em each and every} $u_0^n$, uniformly and independently assign two random bin indices $w_1 \in \llbracket 1,2^{nR_1}\rrbracket$ and $f_1 \in \llbracket 1,2^{n\tilde{R}_1}\rrbracket$.
\item To each pair $(u_0^n,u_j^n)$ for $j=1,2$ uniformly and independently assign random bin index $f'_j \in \llbracket 1,2^{n\tilde{R}'_j}\rrbracket$.
\item To each $v_0^n$ uniformly and independently assign two random bin indices $w_2 \in \llbracket 1,2^{nR_2}\rrbracket$ and $f_2 \in \llbracket 1,2^{n\tilde{R}_2}\rrbracket$.
\item To each pair $(v_0^n,v_j^n)$ for $j=1,2$ uniformly and independently assign random bin index  $f''_j \in \llbracket 1,2^{n\tilde{R}''_j}\rrbracket$.
\item
The random variables representing bin indices are:
\begin{equation}
\quad W_{[1:2]} \;,\quad F_{[1:2]} \;,\quad F'_{[1:2]}\; , \quad F''_{[1:2]}.
\label{eq:shared-randomness}
\end{equation}
\item \sloppy Decoder~1 is a Slepian-Wolf decoder observing $(y_1^n,f_{[1:2]},f'_1,f''_1)$, and producing $(\hatlu_0^n,\hat u_1^n)$ and
$(\hatlv_0^n,\hat v_1^n)$, thus declaring $\hatlw_1=W_1(\hatlu_0^n)$ and $\hatlw_2=W_2(\hatlv_0^n)$ to be the estimate of the pair $(w_1,w_2)$.
\item\sloppy
Decoder~2 is a Slepian-Wolf decoder observing $(y_2^n,f_{[1:2]},f'_2,f''_2)$,
and producing $(\dhatlu_0^n,\dhatlu_2^n)$ and $(\dhatlv_0^n,\dhatlv_2^n)$, thus declaring the bin indices $\dhatlw_1=W_1(\dhatlu_0^n)$ and $\dhatlw_2=W_2(\dhatlv_0^n)$ as the estimate of the pair $(w_1,w_2)$.
\end{itemize}

To condense the notation, we define the following variables:
\begin{gather}
{\bf f} \triangleq \big ( f_{[1:2]},f'_{[1:2]},f''_{[1:2]}\big),
\label{CRF}\\
\hat{\bf u} \triangleq \big (\hatlu_0^n,\dhatlu_0^n,\hatlu_1^n,\dhatlu_2^n,\hatlv_0^n,\dhatlv_0^n,\hatlv_1^n,\dhatlv_2^n\big),
\label{Uhat}
\end{gather}
The random pmf induced by random binning is then as follows:
\begin{align}
&P(u_{[0:2]}^n,v_{[0:2]}^n,x_{[1:2]}^n,y_1^n,y_2^n,z^n,w_{[1:2]},{\bf f}, \hat{\bf u})\nonumber\\
&=p(u_{[0:2]}^n,v_{[0:2]}^n,x_{[1:2]}^n,y_1^n,y_2^n,z^n)P(w_{[1:2]},f_{[1:2]}|u_0^n,v_0^n)\twocolbreaktimes
P(f'_{[1:2]},f''_{[1:2]}|u_{[0:2]}^n,v_{[0:2]}^n)\nonumber\\
&\times
P^{SW}(\hatlu_0^n,\hatlu_1^n,\hatlv_0^n,\hatlv_1^n|y_1^n,f_{[1:2]},f'_1,f''_1)\twocolbreaktimes
P^{SW}(\dhatlu_0^n,\dhatlu_2^n,\dhatlv_0^n,\dhatlv_2^n|y_2^n,f_{[1:2]},f'_2,f''_2)\nonumber\\
&=P(w_{[1:2]},f_{[1:2]},u_0^n,v_0^n)
P(f'_{[1:2]},f''_{[1:2]},u_{[1:2]}^n,v_{[1:2]}^n|u_0^n,v_0^n)\nonumber\\
&\times p(x_1^n|u_{[0:2]}^n) p(x_2^n|v_{[0:2]}^n) p(y_1^n,y_2^n,z^n|x_1^n,x_2^n)\nonumber\\
&\times
P^{SW}(\hatlu_0^n,\hatlu_1^n,\hatlv_0^n,\hatlv_1^n|y_1^n,f_{[1:2]},f'_1,f''_1)\twocolbreaktimes
P^{SW}(\dhatlu_0^n,\dhatlu_2^n,\dhatlv_0^n,\dhatlv_2^n|y_2^n,f_{[1:2]},f'_2,f''_2)\nonumber\\
&=P(w_{[1:2]},f_{[1:2]})P(u_0^n,v_0^n|w_{[1:2]},f_{[1:2]})\twocolbreaktimes
P(f'_{[1:2]},f''_{[1:2]}|u_0^n,v_0^n) P(u_{[1:2]}^n,v_{[1:2]}^n|u_0^n,v_0^n,f'_{[1:2]},f''_{[1:2]})\nonumber\\
&\times p(x_1^n|u_{[0:2]}^n)p(x_2^n|v_{[0:2]}^n) p(y_1^n,y_2^n,z^n|x_1^n,x_2^n)\nonumber\\
&\times
P^{SW}(\hatlu_0^n,\hatlu_1^n,\hatlv_0^n,\hatlv_1^n|y_1^n,f_{[1:2]},f'_1,f''_1)\twocolbreaktimes
P^{SW}(\dhatlu_0^n,\dhatlu_2^n,\dhatlv_0^n,\dhatlv_2^n|y_2^n,f_{[1:2]},f'_2,f''_2).
\label{P}
\end{align}
$P^{SW}$ denotes the pmf of the output of the Slepian-Wolf decoder, which is a random pmf. $\hat{W}_1,\hat{W}_2$ and $\check{W}_1,\check{W}_2$  are omitted because they are functions of other random variables.

We now return to the original problem illustrated in Fig.~\ref{fig}  except that, in addition, a genie provides all terminals with shared randomness described by $(F_{[1:2]},F'_{[1:2]},F''_{[1:2]})$, whose distribution will be clarified in the sequel. In this augmented model:

\begin{itemize}
\item\sloppy The messages $W_1,W_2$ are mutually independent and uniformly distributed with rates $R_1,R_2$ respectively. The shared randomness $(F_1,F_2)$ is uniformly distributed over $\llbracket 1,2^{n\tilde{R}_1}\rrbracket$, $\llbracket 1,2^{n\tilde{R}_2}\rrbracket $, and independent of  $W_1,W_2$.

\item Encoder~1 and~2 are stochastic encoders producing codewords $U_0^n,V_0^n$ according to distributions $P(u_0^n|w_{[1:2]},f_{[1:2]})$ and $P(v_0^n|w_{[1:2]},f_{[1:2]})$, respectively, which are the marginals of distribution $P(u_0^n,v_0^n|w_{[1:2]},f_{[1:2]})$ appearing in~\eqref{P}. This choice of encoder establishes a the connection between the two models.

\item The four random variables $F'_{[1:2]},F''_{[1:2]}$ are mutually independent and uniformly distributed over, respectively, $\llbracket 1,2^{n\tilde{R}'_1}\rrbracket$ and $\llbracket 1,2^{n\tilde{R}'_2}\rrbracket $, $\llbracket 1,2^{n\tilde{R}''_1}\rrbracket$ and $\llbracket 1,2^{n\tilde{R}''_2}\rrbracket $. They are also independent of $(U_0^n,V_0^n)$ and therefore are independent of $(W_{[1:2]},F_{[1:2]})$. 

\item Encoder~1 and~2 further generate $U_{[1:2]}^n$, $V_{[1:2]}^n$ according to $P(u_{[1:2]}^n|u_0^n,v_0^n,f'_{[1:2]},f''_{[1:2]})$ and $P(v_{[1:2]}^n|u_0^n,v_0^n,f'_{[1:2]},f''_{[1:2]})$, respectively, which are marginal distributions of $P(u_{[1:2]}^n,v_{[1:2]}^n|u_0^n,v_0^n,f'_{[1:2]},f''_{[1:2]})$ from~\eqref{P}.

\item Encoder~1 generates $X_1^n$ i.i.d.\ according to $p(x_1|u_{[0:2]})$. Encoder~2 generates $X_2^n$ i.i.d.\ according to $p(x_2|v_{[0:2]})$. $X_1, X_2$ are transmitted over the channel.

\item \sloppy 
Decoders~1 and~2 are Slepian-Wolf decoders inherited from the source model secret key agreement problem, observing respectively $(y_1^n,f_{[1:2]},f'_1,f''_1)$ and $(y_2^n,f_{[1:2]},f'_2,f''_2)$, and producing $(\hatlu_0^n,\hatlu_1^n,\hatlv_0^n,\hatlv_1^n)$ and $(\dhatlu_0^n,\dhatlu_2^n,\dhatlv_0^n,\dhatlv_2^n)$. Therefore the following random PMFs for the decoder output distributions are inherited from the source model:
\begin{align*}
    P^{SW}(\hatlu_0^n,\hatlu_1^n,\hatlv_0^n,\hatlv_1^n|y_1^n,f_{[1:2]},f'_1,f''_1)\\
    P^{SW}(\dhatlu_0^n,\dhatlu_2^n,\dhatlv_0^n,\dhatlv_2^n|y_2^n,f_{[1:2]},f'_2,f''_2)
\end{align*}

\item Decoders~1 and~2 then produce estimates of $(W_1,W_2)$, which are denoted $(\hat{W}_1,\hat{W}_2)$ and $(\check{W}_1,\check{W}_2)$ respectively.
\end{itemize}
The random pmf induced by the random binning and the encoding/decoding strategy is as follows:
\begin{align}
&\hat{P}(u_{[0:2]}^n,v_{[0:2]}^n,y_1^n,y_2^n,z^n,w_{[1:2]},{\bf f}, \hat{\bf u})\nonumber\\
&= p^U(w_{[1:2]})p^U(f_{[1:2]})P(u_0^n,v_0^n|w_{[1:2]},f_{[1:2]})\twocolbreaktimes
p^U(f'_{[1:2]})p^U(f''_{[1:2]})P(u_{[1:2]}^n,v_{[1:2]}^n|u_0^n,v_0^n,f'_{[1:2]},f''_{[1:2]})\nonumber\\
&\times p(x_1^n|u_{[0:2]}^n)p(x_2^n|v_{[0:2]}^n) p(y_1^n,y_2^n,z^n|x_1^n,x_2^n)\nonumber\\
&\times P^{SW}(\hatlu_0^n,\hatlu_1^n,\hatlv_0^n,\hatlv_1^n|y_1^n,f_{[1:2]},f'_1,f''_1)\twocolbreaktimes
P^{SW}(\dhatlu_0^n,\dhatlu_2^n,\dhatlv_0^n,\dhatlv_2^n|y_2^n,f_{[1:2]},f'_2,f''_2),
\label{Phat}
\end{align}
where ${\bf f}$ and $\hat{\bf u}$ are defined in \eqref{CRF} and \eqref{Uhat}, respectively, and $p^U$ is the uniform distribution.



We now find constraints that ensure that the pmfs $\hat{P}$ and $P$ are close in total variation distance. 
For the source model secret key agreement problem, substituting $X_1=X_2 \leftarrow U_0$, and $X_3=X_4 \leftarrow V_0$, in \cite[Theorem~1]{OSRB} implies that $W_{[1:2]}$ is nearly independent of $F_{[1:2]}$ and  $Z^n$, if
\begin{align}
R_1 + \tilde{R}_1 &\leq \ent(U_0|Z),
\label{secur1}\\
R_2 + \tilde{R}_2 &\leq \ent(V_0|Z),
\label{secur2}\\
R_1 + \tilde{R}_1+R_2 + \tilde{R}_2 &\leq \ent(U_0,V_0|Z),
\label{secur3}
\end{align}note that~\cite[Theorem~1]{OSRB} returns a total of 15 inequalities, but the remaining are redundant because of \eqref{secur1}-\eqref{secur3}. The above constraints imply that
\begin{align*}
P(z^n,w_{[1:2]},f_{[1:2]})\approx p(z^n)p^U(w_{[1:2]})p^U(f_{[1:2]}).
\end{align*}
Similarly, substituting
$X_1\leftarrow (U_0,U_1)$, $X_2\leftarrow (U_0,U_2)$, $X_3 \leftarrow (V_0,V_1)$, $X_4\leftarrow (V_0,V_2)$, and $Z\leftarrow (U_0,V_0,Z)$ in \cite[Theorem~1]{OSRB} implies that $(f'_{[1:2]},f''_{[1:2]})$ are nearly mutually independent and independent of  $(U_0,V_0,Z)$, therefore they are independent of $(w_{[1:2]},f_{[1:2]})$, if
\begin{align}
\tilde{R}'_j  &\leq \ent(U_j|U_0,V_0,Z),
\label{secur4}\\
\tilde{R}''_j  &\leq \ent(V_j|U_0,V_0,Z),
\label{secur5}\\
\tilde{R}'_1+\tilde{R}''_j  &\leq \ent(U_1,V_j|U_0,V_0,Z),
\label{secur6}\\
\tilde{R}'_2+\tilde{R}''_j  &\leq \ent(U_2,V_j|U_0,V_0,Z),
\label{secur7}\\
\tilde{R}'_1+\tilde{R}'_2  &\leq \ent(U_1,U_2|U_0,V_0,Z),
\label{secur8}\\
\tilde{R}''_1+\tilde{R}''_2  &\leq \ent(V_1,V_2|U_0,V_0,Z),
\label{secur9}\\
\tilde{R}'_1+\tilde{R}'_2+\tilde{R}''_j  &\leq
\ent(U_1,U_2,V_j|U_0,V_0,Z),
\label{secur10}\\
\tilde{R}'_j+\tilde{R}''_1+\tilde{R}''_2  &\leq
\ent(U_j,V_1,V_2|U_0,V_0,Z),
\label{secur11}\\
\tilde{R}'_1+\tilde{R}'_2+\tilde{R}''_1+\tilde{R}''_2  &\leq \ent(U_1,U_2,V_1,V_2|U_0,V_0,Z),
\label{secur12}
\end{align}for $j=1,2$. The above constraints imply
\begin{align}
P(&z^n,u_0^n,v_0^n,f'_{[1:2]},f''_{[1:2]}) \twocolbreak\approx p(z^n,u_0^n,v_0^n) p^U(f'_{[1:2]})p^U(f''_{[1:2]}).
\end{align}
Hence,
\begin{align}
P(w_{[1:2]},f_{[1:2]})&=\hat P(w_{[1:2]},f_{[1:2]})\twocolbreak=p^U(w_{[1:2]})p^U(f_{[1:2]}),\\
P(f'_{[1:2]},f''_{[1:2]}|u_0^n,v_0^n)&=\hat P(f'_{[1:2]},f''_{[1:2]}|u_0^n,v_0^n)\twocolbreak=p^U(f'_{[1:2]})p^U(f''_{[1:2]}).
\end{align}
In other words, the inequalities \eqref{secur1}-\eqref{secur3} and \eqref{secur4}-\eqref{secur12} imply that
\begin{align}
P(&z^n,w_{[1:2]},f_{[1:2]},f'_{[1:2]},f''_{[1:2]})\twocolbreak\approx
p(z^n)p^U(w_{[1:2]})p^U(f_{[1:2]})p^U(f'_{[1:2]})p^U(f''_{[1:2]}).
\label{PiP}
\end{align}Here, the pmf $P(z^n)$ is equal to $p(z^n)$ because the marginal distribution does not include random binning.

Therefore, the distributions in \eqref{P} and \eqref{Phat} are nearly equal, that is
\begin{align}
P(&u_{[0:2]}^n,v_{[0:2]}^n,y_1^n,y_2^n,z^n,w_{[1:2]},{\bf f}, \hat{\bf u}) \twocolbreak\approx
\hat{P}(u_{[0:2]}^n,v_{[0:2]}^n,y_1^n,y_2^n,z^n,w_{[1:2]},{\bf f}, \hat{\bf u}).
\label{PhatP}
\end{align}

Similar to indirect decoding for channel coding it is possible to use indirect decoding for source coding. More precisely, the first and the second decoders only need $(u_0^n,v_0^n)$ to decode $(w_1,w_2)$. Decoder~1 and Decoder~2 can indirectly decode $(u_0^n,v_0^n)$ from $(y_1^n,f_{[1:2]},f'_1,f''_1)$ and $(y_2^n,f_{[1:2]},f'_2,f''_2)$, respectively. From \cite[Lemma~1]{OSRB} decoding is successful if
\begin{align}
\tilde{R}_1 + \tilde{R}'_j &\geq \ent(U_0,U_j|V_0,V_j,Y_j),
\label{Red41}\\
\tilde{R}_2 + \tilde{R}''_j &\geq \ent(V_0,V_j|U_0,U_j,Y_j),
\label{Red52}\\
\tilde{R}_1 + \tilde{R}'_j + \tilde{R}''_j &\geq \ent(U_0,U_j,V_j|V_0,Y_j),
\label{Red6}\\
\tilde{R}_1 + \tilde{R}_2 + \tilde{R}''_j &\geq \ent(V_0,V_j|U_0,U_j,Y_j),
\label{Red62}\\
\tilde{R}'_j + \tilde{R}_2 + \tilde{R}''_j &\geq \ent(U_j,V_0,V_j|U_0,Y_j),
\label{Red7}\\
\tilde{R}_1 + \tilde{R}'_j + \tilde{R}_2 + \tilde{R}''_j &\geq \ent(U_0,U_j,V_0,V_j|Y_j),
\label{Red8}
\end{align}for $j=1,2$. Note that, inequality \eqref{Red62} is redundant because of \eqref{Red52}.
It yields
\begin{align}
&P(u_{[0:2]}^n,v_{[0:2]}^n,y_1^n,y_2^n,z^n,w_{[1:2]},{\bf f},
  \hat{\bf u}) \twocolbreak
\approx P(u_{[0:2]}^n,v_{[0:2]}^n,y_1^n,y_2^n,z^n,w_{[1:2]},{\bf f})\nonumber\\
&\times\indic{1}\{\hatlu_0^n=\dhatlu_0^n=u_0^n,\hatlu_1^n=u_1^n,\dhatlu_2^n=u_2^n\}\twocolbreaktimes
\indic{1}\{\hatlv_0^n=\dhatlv_0^n=v_0^n,\hatlv_1^n=v_1^n,\dhatlv_2^n=v_2^n\}. 
\label{P1}
\end{align}From Equations \eqref{PhatP}, \eqref{P1}, and the triangle inequality,
\begin{align}
&{\hat P}(u_{[0:2]}^n,v_{[0:2]}^n,y_1^n,y_2^n,z^n,w_{[1:2]},{\bf f},
  \hat{\bf u}) \twocolbreak
\approx P(u_{[0:2]}^n,v_{[0:2]}^n,y_1^n,y_2^n,z^n,w_{[1:2]},{\bf f})\nonumber\\
&\times\indic{1}\{\hatlu_0^n=\dhatlu_0^n=u_0^n,\hatlu_1^n=u_1^n,\dhatlu_2^n=u_2^n\}\twocolbreaktimes
\indic{1}\{\hatlv_0^n=\dhatlv_0^n=v_0^n,\hatlv_1^n=v_1^n,\dhatlv_2^n=v_2^n\}. 
\label{P211}
\end{align}

For convenience, we reintroduce a lemma from~\cite{OSRB}:
\begin{lemma}(\cite[Lemma~4]{OSRB})
\label{lemma:DistEq}
Consider distributions $p_{X^n}$, $p_{Y^n|X^n}$, $q_{X^n}$, and $q_{Y^n|X^n}$ and random pmfs $P_{X^n}$, $P_{Y^n|X^n}$, $Q_{X^n}$, and $Q_{Y^n|X^n}$. Denoting asymptotic equality under total variation with $\approx$, we have:
\begin{enumerate}
    \item 
    \begin{align}
    &P_{X^n} \approx Q_{X^n} \; \Rightarrow \; P_{X^n}P_{Y^n|X^n} \approx \;Q_{X^n}P_{Y^n|X^n} \label{eq:DistLemma1}\\
    &P_{X^n}P_{Y^n|X^n} \approx \;Q_{X^n}Q_{Y^n|X^n}  \; \Rightarrow \; P_{X^n} \approx Q_{X^n}\label{eq:DistLemma2}
    \end{align}
    \item
If $p_{X^n} p_{Y^n|X^n} \approx q_{X^n} q_{Y^n|X^n}$, then there exists a sequence $x^n \in \mathcal{X}^n$ such that 
\begin{equation}
p_{Y^n|X^n=x^n} \approx q_{Y^n|X^n=x^n}.
\label{eq:DistLemma3}
\end{equation}
   \item If $P_{X^n} \approx Q_{X^n}$ and $P_{X^n} P_{Y^n|X^n} \approx P_{X^n}Q_{Y^n|X^n}$, then 
   \begin{equation}
 P_{X^n}P_{Y^n|X^n} \approx Q_{X^n}Q_{Y^n|X^n}.
   \label{eq:DistLemma4}
   \end{equation}
\end{enumerate}
\end{lemma}

Using Lemma~\ref{lemma:DistEq}, Equation~\eqref{eq:DistLemma2}, the marginal distributions of the two sides of \eqref{P211} are asymptotically equivalent, i.e.,
\begin{align}
&{\hat P}(u_{[0:2]}^n,v_{[0:2]}^n,z^n,w_{[1:2]},{\bf f},\hat{\bf u}) \twocolbreak
\approx P(u_{[0:2]}^n,v_{[0:2]}^n,z^n,w_{[1:2]},{\bf f})\nonumber\\
&\times\indic{1}\{\hatlu_0^n=\dhatlu_0^n=u_0^n,\hatlu_1^n=u_1^n,\dhatlu_2^n=u_2^n\}\twocolbreaktimes
\indic{1}\{\hatlv_0^n=\dhatlv_0^n=v_0^n,\hatlv_1^n=v_1^n,\dhatlv_2^n=v_2^n\}. 
\label{P2}
\end{align}
\sloppy
Using Lemma~\ref{lemma:DistEq}, Equation~\eqref{eq:DistLemma1} we multiply the two sides of Equation~\eqref{P2} by the conditional distribution:
\begin{align*}
\hat{P}(&\hatlw_1,\dhatlw_1,\hatlw_2,\dhatlw_2 |u_{[0:2]}^n,v_{[0:2]}^n,z^n,w_{[1:2]},{\bf f},
  \hat{\bf u} ) = \\&\indic{1}\{W_1(\hatlu_0^n)=\hatlw_1,W_1(\dhatlu_0^n)=\dhatlw_1\}\twocolbreaktimes
\indic{1}\{W_2(\hatlv_0^n)=\hatlw_2, W_2(\dhatlv_0^n)=\dhatlw_2\}
\end{align*}
to get:
\begin{align}
&{\hat P}(u_{[0:2]}^n,v_{[0:2]}^n,z^n,w_{[1:2]},{\bf f},
  \hat{\bf u},\hatlw_1,\dhatlw_1,\hatlw_2,\dhatlw_2)\twocolbreak
\approx P(u_{[0:2]}^n,v_{[0:2]}^n,z^n,w_{[1:2]},{\bf f})\nonumber\\
&\times\indic{1}\{\hatlu_0^n=\dhatlu_0^n=u_0^n,\hatlu_1^n=u_1^n,\dhatlu_2^n=u_2^n\}\twocolbreaktimes
\indic{1}\{\hatlv_0^n=\dhatlv_0^n=v_0^n,\hatlv_1^n=v_1^n,\dhatlv_2^n=v_2^n\}\nonumber\\
&\times
\indic{1}\{W_1(\hatlu_0^n)=\hatlw_1,W_1(\dhatlu_0^n)=\dhatlw_1\}\twocolbreaktimes
\indic{1}\{W_2(\hatlv_0^n)=\hatlw_2, W_2(\dhatlv_0^n)=\dhatlw_2\} \nonumber\\
&=P(u_{[0:2]}^n,v_{[0:2]}^n,z^n,w_{[1:2]},{\bf f})\nonumber\\
&\times\indic{1}\{\hatlu_0^n=\dhatlu_0^n=u_0^n,\hatlu_1^n=u_1^n,\dhatlu_2^n=u_2^n\}\twocolbreaktimes
\indic{1}\{\hatlv_0^n=\dhatlv_0^n=v_0^n,\hatlv_1^n=v_1^n,\dhatlv_2^n=v_2^n\}\nonumber\\
&\times\indic{1}\{ \hatlw_1=\dhatlw_1=w_1, \hatlw_2=\dhatlw_2=w_2 \}, 
\label{P3}
\end{align}where $W_1(u_0^n)={\hat w_1}$ and $W_2(v_0^n)={\hat w_2}$ denote the bins assigned to $u_0^n$ and $v_0^n$, respectively. Using \eqref{P3} and Lemma~\ref{lemma:DistEq}, Equation~\eqref{eq:DistLemma1} leads to
\begin{align}
{\hat P}(&z^n,w_{[1:2]},{\bf f},\hatlw_1,\dhatlw_1,\hatlw_2,\dhatlw_2)\approx P(z^n,w_{[1:2]},{\bf f})\twocolbreaktimes
\indic{1}\{ \hatlw_1=\dhatlw_1=w_1, \hatlw_2=\dhatlw_2=w_2 \}, 
\label{P4}
\end{align}

Using Equations \eqref{PiP} and \eqref{P4} and Lemma~\ref{lemma:DistEq}, Equation~\eqref{eq:DistLemma4} leads to
\begin{align}
{\hat P}(&z^n,w_{[1:2]},{\bf f},\hatlw_1,\dhatlw_1,\hatlw_2,\dhatlw_2) \nonumber\\
\approx{}& p(z^n)p^U(w_{[1:2]},f_{[1:2]})p^U(f'_{[1:2]},f''_{[1:2]})\twocolbreaktimes
\indic{1}\{ \hatlw_1=\dhatlw_1=w_1, \hatlw_2=\dhatlw_2=w_2 \}. 
\label{P6}
\end{align}

We now eliminate the shared randomness $(F_{[1:2]},F'_{[1:2]},F''_{[1:2]})$ without affecting the secrecy and reliability requirements. By using Definition~{\ref{pmfDefi}}, Equation~\eqref{P6} ensures that there exists a fixed binning with corresponding pmf $p$ that, if used in place of the random coding strategy $P$ in~\eqref{Phat}, will induce the pmf $\hat{p}$ as follows:
\begin{align}
\hat{p}(&z^n,w_{[1:2]},f_{[1:2]},f'_{[1:2]},f''_{[1:2]},\hatlw_1,\dhatlw_1,\hatlw_2,\dhatlw_2) \nonumber\\
\approx{}&
p(z^n)p^U(w_{[1:2]},f_{[1:2]})p^U(f'_{[1:2]},f''_{[1:2]})\twocolbreaktimes
\indic{1}\{ \hatlw_1=\dhatlw_1=w_1, \hatlw_2=\dhatlw_2=w_2 \}. 
\label{P7}
\end{align}
Now, using Lemma~\ref{lemma:DistEq}, Equation~\eqref{eq:DistLemma3} shows that there exists an instance of $(f_{[1:2]},f'_{[1:2]},f''_{[1:2]})$ such that:
\begin{align}
\hat{p}(&z^n,w_{[1:2]},\hatlw_1,\dhatlw_1,\hatlw_2,\dhatlw_2|f_{[1:2]},f'_{[1:2]},f''_{[1:2]})  \nonumber\\
\approx{}& p(z^n)p^U(w_1)p^U(w_2) \indic{1}\{ \hatlw_1=\dhatlw_1=w_1, \hatlw_2=\dhatlw_2=w_2 \}. 
\label{P71}
\end{align}
This distribution satisfies the secrecy and reliability requirements as follows:

\begin{itemize}
\item Reliability: Using Lemma~\ref{lemma:DistEq}, Equation~\eqref{eq:DistLemma2} leads to
\begin{align}
\hat{p}(&w_{[1:2]},{\hat w_{1,1}},{\hat w_{1,2}},{\hat w_{2,1}},{\hat
    w_{2,2}}|f_{[1:2]},f'_{[1:2]},f''_{[1:2]}) \twocolbreak
\approx \indic{1}\{ \hatlw_1=\dhatlw_1=w_1, \hatlw_2=\dhatlw_2=w_2 \}, 
\label{P8}
\end{align}
\sloppy which is equivalent to:
\begin{gather*}
\hat{p}\Big(\{(\hatuw_1,\hatuw_2) \ne
(W_1,W_2) \} \cup \{(\dhatuw_1,\dhatuw_2) \ne
(W_1,W_2)\}\twocolnewline \bigg| f_{[1:2]},f'_{[1:2]},f''_{[1:2]}\bigg) \to 0.
\end{gather*}
\item Security: Again, using Lemma~\ref{lemma:DistEq}, Equation~\eqref{eq:DistLemma2}
\begin{align}
&{\hat p}(z^n,w_{[1:2]}|f_{[1:2]},f'_{[1:2]},f''_{[1:2]}) \approx p(z^n)p^U(w_1)p^U(w_2). 
\label{P9}
\end{align}
\end{itemize}
Finally, we identify $p(x_1^n|w_1,f_1,f'_{[1:2]})$ and $p(x_2^n|w_2,f_2,f''_{[1:2]})$ (which is done by generating $u_{[0:2]}$ and $v_{[0:2]}$ first, respectively) as encoders and the Slepian-Wolf decoders as decoders for the channel coding problem. These encoders and decoders lead to reliable and secure encoders and decoders.

By applying a computer generated Fourier-Motzkin procedure to \eqref{secur1}-\eqref{secur12}, \eqref{Red41}, \eqref{Red52}, and \eqref{Red8} the achievable rate region for the strong secrecy regime in Theorem~\ref{thstr} is obtained \cite{FMEIT}.
\end{proof}

\begin{remark}
If we assume that \eqref{RAC}, and therefore \eqref{EXCE}, holds, the inequalities \eqref{secur6} for $j=2$, \eqref{secur7} for $j=1$, and \eqref{secur8}-\eqref{secur12} will be redundant and by applying the Fourier-Motzkin procedure \cite{FMEIT,Shirani2011} to \eqref{secur1}-\eqref{secur5}, \eqref{secur6} for $j=1$, \eqref{secur7} for $j=2$, \eqref{Red41}, \eqref{Red52}, and \eqref{Red8} the region in Theorem~\ref{thper} over the distribution \eqref{distachistr} will be achieved. This shows that the region derived by OSRB is a superset of the region derived in the weak secrecy regime.
\end{remark}

\begin{remark}
\sloppy The random distributions $P(u_0^n,v_0^n|w_{[1:2]},f_{[1:2]})$ and
$P(u_{[1:2]}^n,v_{[1:2]}^n|u_0^n,v_0^n,f'_{[1:2]},f''_{[1:2]})$
factorize as $P(u_0^n|w_1,f_1)P(v_0^n|w_2,f_2)$ and
$P(u_{[1:2]}^n|u_0^n,f'_{[1:2]})P(v_{[1:2]}^n|v_0^n,f''_{[1:2]})$, respectively, which means that Encoders~$1$ and $2$ are not using the common randomness and the message available at the other encoder to generate the common and private random variables. The common randomness $(F_1,F'_{[1:2]})$ represents the realization of Encoder~1's codebook and $(F_2,F''_{[1:2]})$ represents the realization of Encoder~2's codebook, which is available at all terminals, but the codebook at one encoder does not depend on the codebook of the other encoder.
\end{remark}

\begin{remark}
The achievable region described in the proof of Theorem~\ref{thstr} was without time sharing, i.e., $Q=\emptyset$. One can incorporate this into the proof by generating i.i.d. copies of $Q$, and sharing it among all terminals and conditioning everything on it.
\end{remark}

\begin{appendices}

\section{Proof of Lemma~\ref{lemma1}}
\label{lemmaproof}
\sloppy 
Let
$N(Q^n,U_0^n,V_0^n,Z^n) = |\{ (k,\ell) \in \llbracket 1,{2^{nS}}\rrbracket \times
\llbracket 1,{2^{nT}}\rrbracket:(Q^n,U_0^n,V_0^n,U_1^n(k),V_1^n(\ell),Z^n) \in
\mathcal{T}_{\epsilon}^{(n)}\} |$. Next, let's define the following error events.

Let $E_1(Q^n,U_0^n,V_0^n,Z^n) = 1$ if $ N(Q^n,U_0^n,V_0^n,Z^n) \ge (1 + {\delta _1}(\epsilon )){2^{n(S+T - \mi(U_1,V_1;Z|Q,U_0,V_0) + \delta (\epsilon ))}}$ and $E_1=0$ otherwise.

Let $E=0$ if $(Q^n,U_0^n,V_0^n,U_1^n(K),V_1^n(L),Z^n) \in \mathcal{T}_{\epsilon}^{(n)}$ and $E_1(Q^n,U_0^n,V_0^n,Z^n,K,L) = 0$, and $E=1$ otherwise. 

We now show that if $S\ge \mi(U_1;Z|Q,U_0,V_0) + \delta (\epsilon)$, $T \ge \mi(V_1;Z|Q,U_0,V_0) + \delta (\epsilon)$, and $S+T \ge \mi(U_1,V_1;Z|Q,U_0,V_0) + \delta (\epsilon)$, then $\Prob( E = 1)  \to 0$ as $n \to \infty$.

By the union bound we have
\begin{align}
\Prob(E = 1)  \le{}& 
\Prob\big( (Q^n,U_0^n,V_0^n,U_1^n(K),V_1^n(L),Z^n) \notin
\mathcal{T}_{\epsilon}^{(n)}\big) \twocolbreak
 +\Prob\big( {E_1}(Q^n,U_0^n,V_0^n,Z^n,K,L) = 1\big).
\end{align}
The first term tends to zero by the main assumption of the Lemma. 

We then partition the event $\{E_1=1\}$ based on the composition of the typical sequences $(Q^n,U_0^n,V_0^n,U_1^n(k),V_1^n(\ell),Z^n) \in \mathcal{T}_{\epsilon}^{(n)}$  :
\begin{itemize}
\item
When all such typical sequences share the same $U_1^n(k)$, i.e., correspond to a single $k$.
\item
When all such typical sequences share the same $V_1^n(\ell)$, i.e., correspond to a single $\ell$.
\item
Neither of the above
\end{itemize}





As usual, each of the three partitioned $E_1$ events gives rise to one rate constraint. We discuss the first in detail; the remaining two follow similarly. Define $A(Q^n,U_0^n,V_0^n,z^n)$ as the event $\{E_1(Q^n,U_0^n,V_0^n,Z^n)=1\}\cap\{Z^n=z^n\}$,
\begin{align}
&\Prob\big(E_1(Q^n,U_0^n,V_0^n,Z^n) = 1\big)\nonumber\\
&= \sum\limits_{(q^n,u_0^n,v_0^n)\in\mathcal{T}_{\epsilon}^{(n)}} \Big[p(q^n)p(u_0^n|q^n)p(v_0^n|q^n)\twocolbreakquad\times\Prob\Big( (E_1(Q^n,U_0^n,V_0^n,Z^n) = 1)\twocolbreakquad|Q^n=q^n,U_0^n=u_0^n,V_0^n = v_0^n\Big)\Big]  \nonumber\\ 
&= \sum\limits_{
\substack{
(q^n,u_0^n,v_0^n) \in \mathcal{T}_{\epsilon}^{(n)}(Q,U_0,V_0) \\
z^n \in\mathcal{T}_{\epsilon}^{(n)}(Z|Q,U_0,V_0)}}
p(q^n)p(u_0^n|q^n)p(v_0^n|q^n)
\twocolbreaktimes
\Prob\big( A(q^n,u_0^n,v_0^n,z^n)|Q^n=q^n,U_0^n=u_0^n,V_0^n = v_0^n\big)  \nonumber\\ 
&\le \sum\limits_{{(q^n,u_0^n,v_0^n)} \in \mathcal{T}_{\epsilon}^{(n)}(Q,U_0,V_0)} {p(q^n)p(u_0^n|q^n)p(v_0^n|q^n)}\twocolbreak\sum\limits_{z^n \in \mathcal{T}_{\epsilon}^{(n)}(Z|Q,U_0,V_0)} 
\Prob\big( ({E_1}(q^n,u_0^n,v_0^n,z^n) = 1)\nonumber\\
&\quad\quad{}|Q^n=q^n,U_0^n=u_0^n,V_0^n = v_0^n\big).
\end{align}


Then,
\begin{align}
& \Prob\big( E_1(q^n,u_0^n,v_0^n,z^n) = 1|Q^n=q^n,U_0^n=u_0^n,V_0^n = v_0^n\big)  = \nonumber\\
& \Prob\big( N(q^n,u_0^n,v_0^n,z^n) \ge (1 + {\delta _1}(\epsilon )){2^{n(T - \mi(V_1;Z|Q,U_0,V_0) + \delta (\epsilon ))}}\big).\nonumber
\end{align}
Define $X_{\ell}=1$ if $(q^n,u_0^n,v_0^n,V_1^n(\ell),z^n) \in \mathcal{T}_{\epsilon}^{(n)}$ and $0$ otherwise. Here, $X_{\ell}$, $\ell \in \llbracket 1,2^{nT}\rrbracket$, are i.i.d.\ Bernoulli-$\alpha$ random variables, where 
\begin{equation}
2^{ - n(\mi(V_1;Z|Q,U_0,V_0) + \delta (\epsilon ))} \le \alpha \le {2^{ - n(\mi(V_1;Z|Q,U_0,V_0) - \delta (\epsilon ))}}\nonumber
\end{equation}
Then
\begin{align}
\Prob\bigg( N({}& q^n,u_0^n,v_0^n,z^n) 
\ge (1 + \delta_1(\epsilon))
{2^{n(T - \mi(V_1;Z|Q,U_0,V_0) + \delta (\epsilon ))}}  
 \twocolbreak\Big|Q^n=q^n,U_0^n = u_0^n,V_0^n = v_0^n \bigg)\nonumber\\ 
 \le \Prob\Bigg({}&  {\sum\limits_{\ell = 1}^{2^{nT}} X_{\ell}  \ge (1 + {\delta _1}(\epsilon )){2^{nT}}\alpha } \twocolbreakonequad\Big|Q^n=q^n, U_0^n = u_0^n,V_0^n = v_0^n \Bigg).\nonumber
\end{align}Applying the Chernoff Bound (e.g., see \cite[Appendix~B]{ElGamalKim}), leads to
\begin{align}
 &\Prob\Bigg( { {\sum\limits_{\ell = 1}^{2^{nT}} X_{\ell}  \ge (1 + \delta _1(\epsilon ))2^{nT}\alpha } \Big|Q^n=q^n,U_0^n = u_0^n,V_0^n = v_0^n} \Bigg) \nonumber\\ 
&\le \exp ( - {2^{nT}}\alpha \delta _1^2(\epsilon )/4) \nonumber\\ 
&\le \exp ( - {2^{n(T - \mi(V_1;Z|Q,U_0,V_0) - \delta (\epsilon ))}}\delta _1^2(\epsilon )/4).
\end{align}Therefore,
\begin{align}
 &\Prob( {E_1}(Q^n,U_0^n,V_0^n,Z^n) = 1)  \nonumber\\ 
 & \le \sum\limits_{(q^n,u_0^n,v_0^n) \in \mathcal{T}_{\epsilon}^{(n)}}{p(q^n)p(u_0^n|q^n)p(v_0^n|q^n)}\twocolbreaktimes
\sum\limits_{z^n \in \mathcal{T}_{\epsilon}^{(n)}(Z|Q,U_0,V_0)} {\exp ( - {2^{n(T - \mi(V_1;Z|Q,U_0,V_0) - \delta (\epsilon ))}}\delta _1^2(\epsilon )/4)}  \nonumber\\ 
 & \le {2^{n\log |\mathcal{Z}|}}\exp ( - {2^{n(T - \mi(V_1;Z|Q,U_0,V_0) - \delta (\epsilon ))}}\delta _1^2(\epsilon )/4),
\end{align}
which tends to zero as $n \to \infty$ if $T \geq \mi(V_1;Z|Q,U_0,V_0) + \delta (\epsilon )$. 

In a similar manner, the bounding of error probability for the second and third partition of $E_1$ (please see above) will give rise to the rate constraints $S \geq \mi(U_1;Z|Q,U_0,V_0) + \delta (\epsilon)$, and  $S+T \geq \mi(U_1,V_1;Z|Q,U_0,V_0) + \delta (\epsilon)$, respectively. Details are ommited for brevity.

Finally, we bound $\ent(L,K|Q^n,U_0^n,V_0^n,Z^n,\mathcal{C})$ as follows:
\begin{align}
& \ent(L,K,E|Q^n,U_0^n,V_0^n,Z^n,C) \nonumber\\ 
& \le 1 + \Prob( E = 1)\ent(L,K|E = 1,Q^n,U_0^n,V_0^n,Z^n,C) \twocolbreak
+ \Prob( E = 0) \ent(L,K|E =0,Q^n,U_0^n,V_0^n,Z^n,C) \nonumber\\ 
& \le 1 + \Prob( E = 1) n(S+T) \twocolbreak
+ \log \big((1 + {\delta _1}(\epsilon )){2^{n(S + T - \mi(U_1,V_1;Z|Q,U_0,V_0) + \delta (\epsilon ))}} \big)
\nonumber\\ 
& \le n(S + T - \mi(U_1,V_1;Z|Q,U_0,V_0) + \delta_2(\epsilon )).
\end{align}

\section{Achievable Rate Region for MAC-WTC Under Randomness Constraint}
\label{thMAWCproof}
\begin{figure}
\centering
\includegraphics[width=\Figwidth]{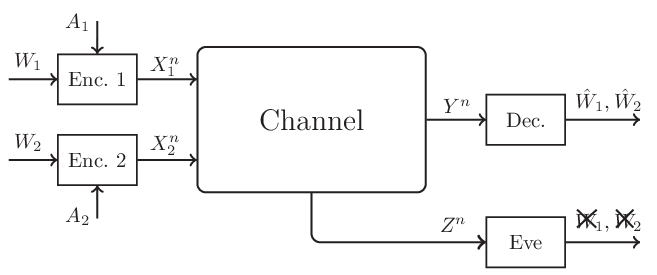}
\caption{Multiple access wiretap channel}

\label{MAWC}
\end{figure}
It is well-known that a stochastic encoding is required to avoid leaking information about the transmitted confidential messages to an eavesdropper.  Here, a new achievability technique for characterizing the trade-off between the rate of the random number to realize the stochastic encoding and the communication rates in multiple access wiretap channel, by employing a variation of superposition coding, is presented.

Consider a MAC-WTC $(\mathcal{X}_1,\mathcal{X}_2,p(y,z|x_1,x_2),\mathcal{Y},\mathcal{Z})$, in which $\mathcal{X}_1$, $\mathcal{X}_2$ are finite input alphabets and $\mathcal{Y}$ and $\mathcal{Z}$ are finite output alphabets at the legitimate receiver and the eavesdropper, respectively (as depicted in Fig.~\ref{MAWC}). In this problem, each transmitter sends a confidential message which is supposed to be decoded by the legitimate receiver and must be kept secret from the eavesdropper. Furthermore, for stochastic encoding, Encoder~1 and Encoder~2 are allowed to use a limited amount of randomness. Thus, we are interested in the trade-off between the rate of randomness, and the rates of confidential messages.
\begin{definition}
\label{defi3}
A $(M_{1,n},M_{2,n},n)$ code for the considered model (Fig.~\ref{MAWC}) consists of the following:
\begin{enumerate}[i)]
\item Two message sets $\mathcal{W}_i=\llbracket 1,M_{i,n}\rrbracket$, $i=1,2$, from which independent messages $W_1$ and $W_2$ are drawn uniformly distributed over their respective sets. Also, Two dummy message sets $\mathcal{A}_i=\llbracket 1,M'_{i,n}\rrbracket$, $i=1,2$, from which independent dummy messages $A_1$ and $A_2$ are drawn uniformly distributed over their respective sets.

\item Deterministic encoders $f_{i,n}$, $i=1,2$, are defined by function $f_{i,n}:\mathcal{W}_i\times\mathcal{A}_i\to\mathcal{X}_i^n$.

\item A decoding function $\phi:\mathcal Y^n\to\mathcal{W}_1\times\mathcal{W}_2$ that assigns $(\hatlw_1,\hatlw_2) \in \llbracket 1,{M_{1,n}}\rrbracket\times\llbracket 1,M_{2,n}\rrbracket$ to received sequence $y^n$.

\end{enumerate}
\end{definition}
The probability of error is given by:
\begin{equation}
\label{pen2}
P_{e} \triangleq\Prob\big(\{(\hatuw_1,\hatuw_2)\ne(w_1,w_2)  \}\big).
\end{equation}
\begin{definition}[\cite{BlochBarros}]
\label{defiperfect2}
A quadruple $(R_1,R_{d_1},R_2,R_{d_2})$ is achievable under weak secrecy if there exists a sequence of $({M_{1,n}},{M_{2,n}},{M'_{1,n}},{M'_{2,n}},n)$ codes with ${M_{1,n}}\ge{2^{n{R_1}}},{M_{2,n}}\ge{2^{n{R_2}}},{M'_{1,n}}\le{2^{n{R_{d_1}}}},{M'_{2,n}}\le{2^{n{R_{d_2}}}}$, so that $P_e
\underset{n\rightarrow\infty}{\xrightarrow{\hspace{0.2in}}} 0$ and
\begin{gather}
\label{Secrecy_Defi2}
\frac{1}{n}\mi(W_1,W_2;Z^n) \underset{n\rightarrow\infty}{\xrightarrow{\hspace{0.2in}}} 0.
\end{gather}
\end{definition}
\begin{theorem}
\label{thperMACWT}
An inner bound on the secrecy capacity region of the multiple access wiretap channel is given by the set of non-negative quadruple $(R_1,R_{d_1},R_2,R_{d_2})$ such that
\begin{align}
R_1 &\le \mi(U;Y|Q,V) - \mi(U;Z|Q),
\label{R11S}\\
R_2 &\le \mi(V;Y|Q,U) - \mi(V;Z|Q),
\label{R12S}\\
R_1 + R_2 &\le \mi(U,V;Y|Q)-\mi(U,V;Z|Q),
\label{SR1S}\\
R_{d_1} &\ge \mi(U;Z|Q) + \mi(X_1;Z|Q,U,V),
\label{Rd1S}\\
R_{d_2} &\ge \mi(V;Z|Q) + \mi(X_2;Z|Q,U,V),
\label{Rd2S}\\
R_{d_1} + R_{d_2} &\ge \mi(U,V;Z|Q) + \mi(X_1,X_2;Z|Q,U,V),
\label{Rd1d2S}
\end{align}
for some
\begin{align}
p(q)p(u|q)p(v|q)p(x_1|u)p(x_2|v)p(y,z|x_1,x_2).
\label{distachiperfMACWT}
\end{align}
\end{theorem}
\begin{remark}
By setting $U=X_1$, $V=X_2$, and by taking sufficiently large $R_{d_1}$ and $R_{d_2}$, the result in Theorem~\ref{thperMACWT} reduces to the achievable rate region of multiple access wiretap channel without common message \cite{GMAWC,GMAWCJamming,YassaeeMAWC}.
\end{remark}
\begin{remark}
By setting $X_2=\emptyset$ and $V=\emptyset$ (or $X_1=\emptyset$ and $U=\emptyset$), the result in Theorem~\ref{thperMACWT} reduces to the capacity rate region of broadcast channel with confidential messages under randomness constraint in \cite[Corollary~11]{OptimalRandomness}.
\end{remark}
\begin{proof}
{\em {Rate Splitting:}}
Divide the dummy message $A_1$ into independent dummy messages $ A_{1,1}\in \llbracket 1, 2^{nR_{1,1}}\rrbracket$ and $A_{1,2}\in \llbracket 1, 2^{nR_{1,2}}\rrbracket$. Also, divide the dummy message $A_2$ into independent dummy messages $A_{2,1}\in \llbracket 1, 2^{nR_{2,1}}\rrbracket$ and $A_{2,2}\in \llbracket 1, 2^{nR_{2,2}}\rrbracket$. Therefore,  $R_{d_1}=R_{1,1}+R_{1,2}$ and $R_{d_2}=R_{2,1}+R_{2,2}$.

{\em {Codebook Generation:}}
Fix $p(q)$, $p(u|q)$, $p(v|q)$, $p(x_1|u)$, $p(x_2|v)$, and $\epsilon>0$.  Randomly and independently generate a typical sequence $q^n$ according to $p(q^n)=\prod\limits_{i = 1}^n p(q_i)$. We suppose that all the terminals know $q^n$.
\begin{enumerate}[i)]

\item Generate $2^{n(R_1+R_{1,1})}$ sequences according to $\prod\nolimits_{i = 1}^n{p_{U|Q}(u_{i}|q_i)}$. Then, randomly bin these $2^{n(R_1+R_{1,1})}$ sequences into $2^{nR_1}$ bins. We index these sequences as $u^n(w_1,a_{1,1})$. For each $(w_1,a_{1,1})$, generate $2^{nR_{1,2}}$ codewords $x_1^n(w_1,a_{1,1},a_{1,2})$ each according to $\prod\nolimits_{i = 1}^n{p_{X_1|U}(x_{1,i}|u_{i})}$.

\item Generate $2^{n(R_2+R_{2,1})}$ sequences according to $\prod\nolimits_{i = 1}^n{p_{V|Q}(v_{i}|q_i)}$. Then, randomly bin these $2^{n(R_2+R_{2,1})}$ sequences into $2^{nR_2}$ bins. We index these sequences as $v^n(w_2,a_{2,1})$. For each $(w_2,a_{2,1})$, generate $2^{nR_{2,2}}$ codewords $x_1^n(w_2,a_{2,1},a_{2,2})$ each according to $\prod\nolimits_{i = 1}^n{p_{X_2|V}(x_{2,i}|v_{i})}$.

\end{enumerate}

{\em {Encoding:}} 
To send the message $w_1$, the Encoder~1 splits $a_1$ into $(a_{1,1}, a_{1,2})$, and chooses $u^n(w_1,a_{1,1})$. Then it chooses codeword $x_1^n(w_1,a_{1,1},a_{1,2})$ and send it over the channel.

To send the message $w_2$, the Encoder~2 splits $a_2$ into $(a_{2,1}, a_{2,2})$, and chooses $v^n(w_2,a_{2,1})$. Then it chooses codeword $x_2^n(w_2,a_{2,1},a_{2,2})$ and send it over the channel.


{\em {Decoding and Error Probability Analysis:}}
\begin{itemize}
\item \sloppy Decoder decodes $(w_1,w_2)$ by finding a unique pair $(w_1,w_2)$ such that $(q^n,u^n(w_1,a_{1,1}),v^n(w_2,a_{2,1}),y^n)\in\mathcal{T}_{\epsilon}^{(n)}({p_{U,V,Y}})$ for some $(a_{1,1},a_{2,1})$. The probability of error for Receiver goes to zero as $n\rightarrow\infty$ if we choose \cite{ElGamalKim}
\begin{align}
&R_1 + R_{1,1} \leq \mi(U;Y|Q,V) - \epsilon,
\label{dec1s}\\
&R_2 + R_{2,1}\leq \mi(V;Y|Q,U) - \epsilon,
\label{dec2s}\\
&R_1 + R_{1,1} + R_2 + R_{2,1} \leq \mi(U,V;Y|Q) - \epsilon.
\label{dec12s}
\end{align}

\end{itemize}

{\em {Equivocation Calculation:}}
We analyze mutual information between $(W_1,W_2)$ and $Z^n$, averaged over all random codebooks
\begin{align}
\mi(&W_1,W_2;Z^n|Q^n,\mathcal{C})\nonumber\\
={}&\mi(W_1,W_2,A_{1,1},A_{1,2},A_{2,1},A_{2,2};Z^n|Q^n,\mathcal{C})\twocolbreak 
-\mi(A_{1,1},A_{1,2},A_{2,1},A_{2,2};Z^n|W_1,W_2,Q^n,\mathcal{C})\nonumber\\ 
\mathop= \limits^{(a)}{}&\mi(W_1,W_2,A_{1,1},A_{1,2},A_{2,1},A_{2,2},X_1^n,X_2^n;Z^n|Q^n,\mathcal{C})\twocolbreak 
-\mi(A_{1,1},A_{1,2},A_{2,1},A_{2,2};Z^n|W_1,W_2,Q^n,\mathcal{C})\nonumber\\ 
\mathop= \limits^{(b)}{}& \mi(X_1^n,X_2^n;Z^n|Q^n,\mathcal{C})\twocolbreak -\mi(A_{1,1},A_{1,2},A_{2,1},A_{2,2};Z^n|W_1,W_2,Q^n,\mathcal{C})\nonumber\\ 
={}&\mi(X_1^n,X_2^n;Z^n|Q^n,\mathcal{C}) - \mi(A_{1,1},A_{2,1};Z^n|W_1,W_2,Q^n,\mathcal{C}) \nonumber\\
&-\mi(A_{1,2},A_{2,2};Z^n|W_1,W_2,A_{1,1},A_{1,2},Q^n,\mathcal{C})\nonumber\\ 
={}&\mi(X_1^n,X_2^n;Z^n|Q^n,\mathcal{C}) - \ent(A_{1,1},A_{2,1}|W_1,W_2,Q^n,\mathcal{C}) \twocolbreak
+ \ent(A_{1,1},A_{2,1}|W_1,W_2,Z^n,Q^n,\mathcal{C})\nonumber\\
&-\ent(A_{1,2},A_{2,2}|W_1,W_2,A_{1,1},A_{2,1},Q^n,\mathcal{C}) \twocolbreak
+ \ent(A_{1,2},A_{2,2}|W_1,W_2,A_{1,1},A_{2,1},Z^n,Q^n,\mathcal{C})
\label{EquivocationS}
\end{align}
where $(a)$ is due to $X_1^n$ and $X_2^n$ are deterministic functions of $(W_1,A_{1,1},A_{1,2})$ and $(W_2,A_{2,1},A_{2,2})$, respectively. Also, $(b)$ is due to the fact that, given $X_1^n$ and $X_2^n$, the indices $W_1$, $W_2$, $A_{1,1}$, $A_{1,2}$ ,$A_{2,1}$, and $A_{2,2}$ are uniquely determined.

The first term in \eqref{EquivocationS} is bounded as:
\begin{align}
\mi(X_1^n,X_2^n;Z^n|Q^n,\mathcal{C}) \leq n\mi(X_1,X_2;Z|Q) + n\epsilon,
\label{firstterms}
\end{align}
where $\epsilon \underset{n\rightarrow\infty}{\xrightarrow{\hspace{0.2in}}} 0$ similar to \cite{ElGamalKim}.

For the second term in \eqref{EquivocationS} we have
\begin{align}
\ent(A_{1,1},A_{2,1}|W_1,W_2,Q^n,\mathcal{C})=n(R_{1,1} + R_{2,1}).
\label{secterms}
\end{align}
For the third term, substituting $U_0\leftarrow Q$, $V_0\leftarrow Q$, $U_1\leftarrow U$, and $V_1\leftarrow V$ in Lemma~\ref{lemma1} result that if $\Prob\big((Q^n,U^n(W_1,A_{1,1}),V^n(W_2,A_{2,1}),Z^n) \in\mathcal{T}_{\epsilon}^{(n)}\big)\underset{n\rightarrow\infty}{\xrightarrow{\hspace{0.2in}}} 1$ and 
\begin{align}
R_{1,1} &\geq \mi(U;Z|Q) + \epsilon,
\label{condition1}\\
R_{2,1} &\geq \mi(V;Z|Q) + \epsilon,
\label{condition2}\\
R_{1,1} + R_{2,1} &\geq \mi(U,V;Z|Q) + \epsilon.
\label{condition3}
\end{align}
Then,
\begin{align}
&\ent(A_{1,1},A_{2,1}|W_1,W_2,Z^n,Q^n,\mathcal{C}) \twocolbreak
\le  n(R_{1,1} + R_{2,1} - \mi(U,V;Z|Q) + \epsilon).
\label{thidterm1s}
\end{align}
Here, this condition holds because
\begin{align}
\Prob\big({}&(Q^n,U^n(W_1,A_{1,1}),X_1^n(W_1,A_{1,1},A_{1,2}),\twocolbreak V^n(W_2,A_{2,1}),X_2^n(W_2,A_{2,1},A_{2,2}),Z^n)\in \mathcal{T}_{\epsilon}^{(n)}\big)  \underset{n\rightarrow\infty}{\xrightarrow{\hspace{0.2in}}} 1.
\label{typicalitys}
\end{align}
Now, we bound the fourth term in \eqref{EquivocationS},
\begin{align}
\ent(A_{1,2},A_{2,2}|W_1,W_2,A_{1,1},A_{2,1},Q^n,\mathcal{C}) = n(R_{1,2} + R_{2,2}).
\label{FourthTerms}
\end{align}
Now, we bound the last term in \eqref{EquivocationS} by applying Lemma~\ref{lemma1},
\begin{align}
&\ent(A_{1,2},A_{2,2}|W_1,W_2,A_{1,1},A_{2,1},Z^n,Q^n,\mathcal{C})\twocolbreak
\leq n(R_{1,2} + R_{2,2} - \mi(X_1,X_2;Z|Q,U,V)+\epsilon),
\label{fourthterms2}
\end{align}if \eqref{typicalitys} holds and
\begin{align}
    R_{1,2} &\geq \mi(X_1;Z|Q,U,V)+\epsilon,
    \label{SecCondition1}\\
    R_{2,2} &\geq \mi(X_2;Z|Q,U,V)+\epsilon,
    \label{SecCondition2}\\
    R_{1,2} + R_{2,2} &\geq \mi(X_1,X_2;Z|Q,U,V)+\epsilon.
    \label{SecCondition3}
\end{align}

Substituting \eqref{firstterms}, \eqref{secterms}, \eqref{thidterm1s}, \eqref{FourthTerms}, and \eqref{fourthterms2} into \eqref{EquivocationS} yields
\begin{align}
&\mi(W_1,W_2;Z^n|Q^n,\mathcal{C}) \leq n\mi(X_1,X_2;Z|Q) - n(R_{1,1} + R_{2,1}) \nonumber\\
& + n(R_{1,1} + R_{2,1} - \mi(U,V;Z|Q) + \epsilon) - n(R_{1,2} + R_{2,2})\nonumber\\
& + n(R_{1,2} + R_{2,2} - \mi(X_1,X_2;Z|Q,U,V)+\epsilon).
\end{align}
Therefore $\mi(W_1,W_2;Z^n|Q^n,\mathcal{C})\leq 2n\epsilon$. 
By applying the Fourier-Motzkin procedure \cite{FMEIT} to \eqref{dec1s}-\eqref{dec12s}, \eqref{condition1}-\eqref{condition3}, \eqref{SecCondition1}-\eqref{SecCondition3}, $R_{d_1}=R_{1,1}+R_{1,2}$, and $R_{d_2}=R_{2,1}+R_{2,2}$ we obtain the region in Theorem~\ref{thperMACWT}.
\end{proof}

\section{Proof of Theorem~\ref{thper}}
\label{thperproof}
The coding scheme is based on superposition coding, Wyner's random binning \cite{Wyner}, Marton coding, and applying indirect decoding \cite{ChiaElGamal}. 

The random code generation is as follows:

Fix $p(q)$, $p(u_0|q)$, $p(u_1,u_2|u_0)$, $p(v_0|q)$, $p(v_1,v_2|v_0)$, $p(x_1|u_0,u_1,u_2)$, $p(x_2|v_0,v_1,v_2)$, $\epsilon_1<\min\{\epsilon',\epsilon''\}$, and $\epsilon_2<\min\{\epsilon',\epsilon''\}$.\\
{\em {Codebook Generation:}} 
Randomly and independently generate a typical sequence $q^n$ according to $p(q^n)=\prod\limits_{i = 1}^n p(q_i)$. We suppose that all the terminals know $q^n$.
\begin{enumerate}[i)]
\item Generate $2^{n\tilde{R}_1}$ codewords $u_0^n(\ell_0)$ each according to $\prod\nolimits_{i = 1}^n{p_{U_0|Q}(u_{0,i}|q_i)}$. Then, randomly bin the $2^{n\tilde{R}_1}$ codewords into $2^{nR_1}$ bins, $\mathcal{B}(w_1)$, $w_1 \in \llbracket 1,2^{nR_1}\rrbracket$. For each $\ell_0$, generate ${2^{n\rho_1}}$ codewords $u_1^n(\ell_0,t_1)$ each according to $\prod\nolimits_{i = 1}^n {p_{U_1|U_0}(u_{1,i}|u_{0,i})}$. Then, randomly bin the $2^{n\rho_1}$ codewords into $2^{n\rho'_1}$ bins, $\mathcal{B}(\ell_0,\ell_1)$, $\ell_1 \in \llbracket 1,2^{n{\rho'_1}}\rrbracket$. Similarly, for each $\ell_0$, generate ${2^{n\tilde{\rho}_1}}$ codewords $u_2^n(\ell_0,t_2)$ each according to $\prod\nolimits_{i = 1}^n {p_{U_2|U_0}(u_{2,i}|u_{0,i})}$. Then, randomly bin the ${2^{n\tilde{\rho}_1}}$ codewords into $2^{n{\tilde{\rho}'_1}}$ bins, $\mathcal{B}(\ell_0,\ell_2)$, $\ell_2 \in \llbracket 1,2^{n{\tilde{\rho}'_1}}\rrbracket$.
\item Similarly, generate $2^{n\tilde{R}_2}$ codewords $v_0^n(\ell'_0)$ each according to $\prod\nolimits_{i = 1}^n {p_{V_0|Q}(v_{0,i}|q_i)}$. Then, randomly bin the $2^{n\tilde{R}_2}$ codewords into $2^{nR_2}$ bins,   $\mathcal{B}(w_2)$, $w_2 \in \llbracket 1,2^{nR_2}\rrbracket$. For each $\ell'_0$, generate $2^{n\rho_2}$ codewords $v_1^n(\ell'_0,s_1)$ each according to $\prod\nolimits_{i = 1}^n {p_{V_1|V_0}(v_{1,i}|v_{0,i})}$. Then, randomly bin the $2^{n\rho_2}$ codewords into $2^{n\rho'_2}$ bins, $\mathcal{B}(\ell'_0,\ell'_1)$, $\ell'_1 \in\llbracket 1,2^{n\rho'_2}\rrbracket$. Similarly, for each $\ell'_0$, generate $2^{n\tilde{\rho}_2}$ codewords $v_2^n(\ell'_0,s_2)$ each according to $\prod\nolimits_{i = 1}^n {p_{V_2|V_0}(v_{2,i}|v_{0,i})}$. Then, randomly bin the $2^{n\tilde{\rho}_2}$ codewords into $2^{n\tilde{\rho}'_2}$ bins, $\mathcal{B}(\ell'_0,\ell'_2)$, $\ell'_2 \in\llbracket 1,2^{n\tilde{\rho}'_2}\rrbracket$.
\end{enumerate}
{\em {Encoding:}} To send the message $w_1$, the encoder $f_1$ first uniformly chooses index $L_0 \in \mathcal{B}(w_1)$. Then, it uniformly chooses a pair of indices $(L_1,L_2)$ and selects a jointly typical sequence pair $(u_1^n(L_0,t_1(L_0,L_1)),u_2^n(L_0,t_2(L_0,L_1)))\in\mathcal{T}_{\epsilon_1}^{(n)}(U_1,U_2|U_0)$ in the product bin. If the encoder $f_1$ finds more than one such pair, then it chooses one of them uniformly at random. We have an error if there is no such pair, in which the encoder $f_1$ uniformly at random chooses $t_1\in \mathcal{B}(L_0,L_1)$, $t_2 \in \mathcal{B}(L_0,L_2)$. The error
probability of the last event approaches to zero as $n \to \infty$, if
\cite{ElGamalVanMeulen}
\begin{align}
\rho'_1 + \tilde{\rho}'_1 \leq \rho_1 + \tilde{\rho}_1 - \mi(U_1;U_2|U_0) - \epsilon_1.
\label{dec01}
\end{align}
\sloppy Finally, the encoder $f_1$ generates a sequence $X_1^n$ at random
according to $\prod\nolimits_{i=1}^n{p(x_{1,i}|u_{0,i},u_{1,i},u_{2,i})}$. Encoder~2 proceeds similarly to encode $w_2$ and sends codeword $X_2^n$. The probability of not finding a jointly typical sequence pair
$(v_1^n(L'_0,s_1(L'_0,L'_1)),v_2^n(L'_0,s_2(L'_0,L'_1)))\in\mathcal{T}_{\epsilon_2}^{(n)}(V_1,V_2|V_0)$ in the product bin approaches to zero as $n \to \infty$, if \cite{ElGamalVanMeulen}
\begin{align}
\rho'_2 + \tilde{\rho}'_2 \leq \rho_2 + \tilde{\rho}_2 - \mi(V_1;V_2|V_0) - \epsilon_2.
\label{dec02}
\end{align}
\begin{figure}
\centering
\includegraphics[width=\Figwidth]{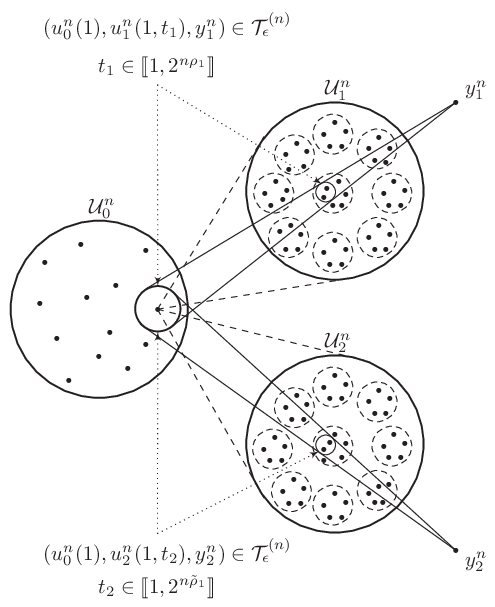}
\caption{Codebook structure and indirect decoding for $u_0^n(1)$ via $u_1^n(1,t_1)$ and $u_2^n(1,t_2)$ for the situation that there is just one transmitter.}
\label{fig6}
\end{figure}
{\em {Decoding and Error Probability Analysis:}}
\begin{itemize}
\item Let $(W_1,L_0,T_1)$ and $(W_2,L'_0,S_1)$ denote the transmitted indices by the first and second transmitter, respectively, and let $(\hat{W}_1,\hat{L}_0,\hat{T}_1)$ and $(\hat{W}_2,\hat{L}'_0,\hat{S}_1)$ denote the corresponding decoded messages by the first receiver, respectively. Receiver~1 decodes $(L_0,L'_0)$ and therefore $(w_1,w_2)$ indirectly by finding a unique pair $(\hatll_0,\hatll'_0)$ such that $(q^n,u_0^n(\hatll_0),u_1^n(\hatll_0,t_1),v_0^n(\hatll'_0),v_1^n(\hatll'_0,s_1),y_1^n)\in\mathcal{T}_{\epsilon'}^{(n)}(U_0,U_1,V_0,V_1,Y_1)$ for some $t_1 \in \llbracket 1,2^{n\rho_1}\rrbracket$ and $s_1 \in \llbracket 1,2^{n\rho_2}\rrbracket$. The idea of indirect decoding for the situation that there is just one transmitter is shown in Fig.~\ref{fig6}. The error event $(\hat{W}_1,\hat{W}_1)\ne (W_1,W_1)$ occurs only if at least one of the following events occurs:
\begin{align}
    \mathcal{E}_1 &= \Big\{ \big(Q^n,U_0^n(\ell_0),U_1^n(\ell_0,t_1),\twocolbreakonequad V_0^n(\ell'_0),V_1^n(\ell'_0,s_1),Y_1^n\big) \notin \mathcal{T}_\epsilon ^{(n)}\Big\},\label{eq:Error_Event_1}\\
    \mathcal{E}_2 &= \Big\{ \big(Q^n,U_0^n(\hat{\ell}_0),U_1^n(\hat{\ell}_0,\hat{t}_1), V_0^n(\ell'_0), V_1^n(\ell'_0,s_1),\twocolbreakonequad Y_1^n\big) \in \mathcal{T}_\epsilon ^{(n)}\,\,\mbox{for some}\,\,\hat{\ell}_0\ne\ell_0,\hat{t}_1 \Big\},\label{eq:Error_Event_2}\\
    \mathcal{E}_3 &= \Big\{ \big(Q^n,U_0^n(\ell_0),U_1^n(\ell_0,t_1),V_0^n(\hat{\ell}'_0),V_1^n(\hat{\ell}'_0,\hat{s}_1),\twocolbreakonequad Y_1^n\big) \in \mathcal{T}_\epsilon ^{(n)}\,\,\mbox{for some}\,\, \hat{\ell}'_0\ne \ell'_0, \hat{s}_1 \Big\},\label{eq:Error_Event_3}\\
    \mathcal{E}_4 &= \Big\{ \big(Q^n,U_0^n(\hat{\ell}_0),U_1^n(\hat{\ell}_0,\hat{t}_1),V_0^n(\ell'_0),V_1^n(\ell'_0,\hat{s}_1),\nonumber\\
    &\quad Y_1^n\big) \in \mathcal{T}_\epsilon ^{(n)}\,\,\mbox{for some}\,\,\hat{\ell}_0\ne\ell_0,\hat{t}_1, \hat{s}_1\ne s_1 \Big\},\label{eq:Error_Event_4}\\
    \mathcal{E}_5 &= \Big\{ \big(Q^n,U_0^n(\ell_0),U_1^n(\ell_0,\hat{t}_1),V_0^n(\hat{\ell}'_0),V_1^n(\hat{\ell}'_0,\hat{s}_1),\nonumber\\
    &\quad Y_1^n\big) \in \mathcal{T}_\epsilon ^{(n)}\,\,\mbox{for some}\,\,\hat{t}_1\ne t_1, \hat{\ell}'_0\ne \ell'_0, \hat{s}_1 \Big\},\label{eq:Error_Event_5}\\
    \mathcal{E}_6 &= \Big\{ \big(Q^n,U_0^n(\hat{\ell}_0),U_1^n(\hat{\ell}_0,\hat{t}_1),V_0^n(\hat{\ell}'_0),V_1^n(\hat{\ell}'_0,\hat{s}_1)\nonumber\\
    &\quad,Y_1^n\big) \in \mathcal{T}_\epsilon ^{(n)}\,\,\mbox{for some}\,\,\hat{\ell}_0\ne\ell_0,\hat{t}_1, \hat{\ell}'_0\ne \ell'_0, \hat{s}_1 \Big\}.\label{eq:Error_Event_6}
\end{align}Therefore, by Union Bound the average probability of error for decoder~1 is upper bounded as
\begin{align*}
    P_{e_1} \leq\Prob(\mathcal{E}_1)+\Prob(\mathcal{E}_2)+\Prob(\mathcal{E}_3)+\Prob(\mathcal{E}_4)+\Prob(\mathcal{E}_5)+\Prob(\mathcal{E}_6).
\end{align*} By law of large numbers, $\Prob(\mathcal{E}_1)$ tends to zero as $n\to\infty$. By packing lemma \cite[Lemma~3.1]{ElGamalKim} $\Prob(\mathcal{E}_2)$ to $\Prob(\mathcal{E}_6)$ respectively tend to zero as $n\to\infty$ if
\begin{align}
\twocolAlignMarker\tilde{R}_1 + \rho_1 \onecolAlignMarker< \mi(U_0,U_1;Y_1|Q,V_0,V_1),
\label{eq:dec11}\\
\twocolAlignMarker\tilde{R}_2 + \rho_2 \onecolAlignMarker< \mi(V_0,V_1;Y_1|Q,U_0,U_1),
\label{eq:dec12}\\
\twocolAlignMarker\tilde{R}_1 + \rho_1 + \rho_2 \onecolAlignMarker< \mi(U_0,U_1,V_1;Y_1|Q,V_0),
\label{eq:dec13}\\
\twocolAlignMarker \rho_1 + \tilde{R}_2 + \rho_2 \onecolAlignMarker< \mi(U_1,V_0,V_1;Y_1|Q,U_0),
\label{eq:dec14}\\
\twocolAlignMarker\tilde{R}_1 + \rho_1 + \tilde{R}_2 + \rho_2 \onecolAlignMarker< \mi(U_0,U_1,V_0,V_1;Y_1|Q).
\label{dec15}
\end{align}

\item Similarly Receiver~2 decodes $(L_0,L'_0)$ and therefore $(w_1,w_2)$ indirectly by finding a unique pair $(\dhatll_0,\dhatll'_0)$ such that $(q^n,u_0^n(\dhatll_0),u_2^n(\dhatll_0,t_2),v_0^n(\dhatll'_0),v_2^n(\dhatll'_0,s_2),y_2^n) \in\mathcal{T}_{\epsilon''}^{(n)}(U_0,U_2,V_0,V_2,Y_2)$ for some $t_2 \in \llbracket 1,2^{n\tilde{\rho}_1}\rrbracket$ and $s_2 \in \llbracket 1,2^{n\tilde{\rho}_2}\rrbracket$. The error analysis for the second receiver is similar to the first receiver and for the interest of brevity it is omitted here. Similar to Receiver~1 the The probability of error for Receiver~2 goes to zero as $n\rightarrow\infty$ if we choose \cite{ElGamalKim}
\begin{align}
\twocolAlignMarker\tilde{R}_1 + \tilde{\rho}_1 \onecolAlignMarker< \mi(U_0,U_2;Y_2|Q,V_0,V_2),
\label{eq:dec21}\\
\twocolAlignMarker\tilde{R}_2 + \tilde{\rho}_2 \onecolAlignMarker< \mi(V_0,V_2;Y_2|Q,U_0,U_2),
\label{eq:dec22}\\
\twocolAlignMarker\tilde{R}_1 + \tilde{\rho}_1 + \tilde{\rho}_2 \onecolAlignMarker< \mi(U_0,U_2,V_2;Y_2|Q,V_0),
\label{eq:dec23}\\
\twocolAlignMarker \tilde{\rho}_1 + \tilde{R}_2 + \tilde{\rho}_2 \onecolAlignMarker< \mi(U_2,V_0,V_2;Y_2|Q,U_0) ,
\label{eq:dec24}\\
\twocolAlignMarker\tilde{R}_1 + \tilde{\rho}_1 + \tilde{R}_2 + \tilde{\rho}_2 \onecolAlignMarker< \mi(U_0,U_2,V_0,V_2;Y_2|Q).
\label{eq:dec25}
\end{align}
\end{itemize}
{\em {Equivocation Calculation:}}
We analyze mutual information between $(W_1,W_2)$ and $Z^n$, averaged over all random codebooks
\begin{align}
\mi(&W_1,W_2;Z^n|Q^n,\mathcal{C})\nonumber\\
={}&\mi(W_1,W_2,L_0,T_1,T_2,L'_0,S_1,S_2;Z^n|Q^n,\mathcal{C})\twocolbreak 
-\mi(L_0,T_1,T_2,L'_0,S_1,S_2;Z^n|W_1,W_2,Q^n,\mathcal{C})\nonumber\\ 
\le{}&
\mi(U_0^n,U_1^n,U_2^n,V_0^n,V_1^n,V_2^n;Z^n|Q^n,\mathcal{C})\twocolbreak
-\mi(L_0,L'_0;Z^n|W_1,W_2,Q^n,\mathcal{C})\nonumber\\ 
&-\mi(T_1,T_2,S_1,S_2;Z^n|L_0,L'_0,Q^n,\mathcal{C})\nonumber\\ 
={}&\mi(U_0^n,U_1^n,U_2^n,V_0^n,V_1^n,V_2^n;Z^n|Q^n,\mathcal{C})\twocolbreak
- \ent(L_0,L'_0|W_1,W_2,Q^n,\mathcal{C})\nonumber\\
&+\ent(L_0,L'_0|Z^n,W_1,W_2,Q^n,\mathcal{C})\twocolbreak
 -\mi(T_1,T_2, S_1,S_2;Z^n|L_0,L'_0,Q^n,\mathcal{C}),
\label{Equivocation}
\end{align}
where the inequality is due to the data processing inequality. Here, $T_1$, $T_2$, $S_1$, and $S_2$ are deterministic functions of $(L_0,L_1)$, $(L_0,L_2)$, $(L'_0,L'_1)$, and $(L'_0,L'_2)$, respectively.

The first term in \eqref{Equivocation} is bounded as:
\begin{align}
&\mi(U_0^n,U_1^n,U_2^n,V_0^n,V_1^n,V_2^n;Z^n|Q^n,\mathcal{C})\twocolbreak
\leq n\mi(U_0,U_1,U_2,V_0,V_1,V_2;Z|Q) + n\epsilon,
\label{firstterm}
\end{align}
as $n \to \infty$ where $\epsilon \to 0$ \cite{ElGamalKim}.

For the second term in \eqref{Equivocation} we have
\begin{equation}
\label{secterm}
\ent(L_0,L'_0|W_1,W_2,Q^n,\mathcal{C})=n(\tilde{R}_1 - R_1 + \tilde{R}_2 - R_2).
\end{equation}
For the third term, substituting $U_0\leftarrow Q$, $V_0\leftarrow Q$, $U_1\leftarrow U_0$, and $V_1\leftarrow V_0$ in Lemma~\ref{lemma1} result that,
\begin{align}
&\ent(L_0,L'_0|Z^n,W_1,W_2,Q^n,\mathcal{C}) \twocolbreak
\le  n(\tilde{R}_1 - R_1 + \tilde{R}_2 - R_2 - \mi(U_0,V_0;Z|Q) + \epsilon),
\label{thidterm1}
\end{align}if $\Prob\big((Q^n,U_0^n(L_0),V_0^n(L'_0),Z^n) \in\mathcal{T}_{\epsilon}^{(n)}\big)\to 1$ as $n\to \infty$ and $\tilde{R}_1
- R_1 \geq \mi(U_0;Z|Q) + \epsilon$, $\tilde{R}_2 - R_2 \geq
\mi(V_0;Z|Q) + \epsilon$, and $\tilde{R}_1 - R_1 + \tilde{R}_2 - R_2 \geq \mi(U_0,V_0;Z|Q) + \epsilon$.

Here, the first condition holds because
\begin{align}
\Prob\big({}&(Q^n,U_0^n(L_0),U_1^n(L_0,t_1(L_0,L_1)),U_2^n(L_0,t_2(L_0,L_1))\twocolbreak,V_0^n(L'_0),V_1^n(L'_0,s_1(L'_0,L'_1))\nonumber\\
&,V_2^n(L'_0,s_2(L'_0,L'_1)),Z^n)\in \mathcal{T}_{\epsilon}^{(n)}\big)\to 1\label{typicality}
\end{align}
as $n\to \infty$. Now, we bound the last term in \eqref{Equivocation}
\begin{align}
\mi(&T_1,T_2,S_1,S_2;Z^n|L_0,L'_0,Q^n,\mathcal{C})\nonumber\\
={}&\ent(T_1,T_2,S_1,S_2|L_0,L'_0,Q^n,\mathcal{C})\twocolbreak
-\ent(T_1,T_2,S_1,S_2|Z^n,L_0,L'_0,Q^n,\mathcal{C})\nonumber\\
\mathop  =
\limits^{(a)}{}&\ent(T_1,T_2,S_1,S_2,L_1,L_2,L'_1,L'_2|L_0,L'_0,Q^n,\mathcal{C})\twocolbreak
-\ent(T_1,T_2,S_1,S_2|Z^n,L_0,L'_0,Q^n,\mathcal{C})\nonumber\\
\geq{}&
\ent(L_1,L_2,L'_1,L'_2|L_0,L'_0,Q^n,\mathcal{C})\twocolbreak
-\ent(T_1,S_1|Z^n,L_0,L'_0,Q^n,\mathcal{C})\twocolbreak
-\ent(T_2,S_2|Z^n,L_0,L'_0,Q^n,\mathcal{C})\nonumber\\
\mathop  = \limits^{(b)}{}& \ent(L_1,L_2|L_0,L'_0,Q^n,\mathcal{C})+\ent(L'_1,L'_2|L_0,L'_0,Q^n,\mathcal{C})\nonumber\\
&-\ent(T_1,S_1|Z^n,L_0,L'_0,Q^n,\mathcal{C})\twocolbreak
-\ent(T_2,S_2|Z^n,L_0,L'_0,Q^n,\mathcal{C}),
\label{IT}
\end{align}
where $(a)$ is due to given the codebook~$\mathcal{C}$ and $(L_0,L'_0)$, $(L_1,L_2,L'_1,L'_2)$ is a deterministic function of $(T_1(L_0,L_1),T_2(L_0,L_2),S_1(L'_0,L'_1),S_2(L'_0,L'_2))$, and $(b)$ holds due to the fact that given $(L_0,L'_0,Q^n,\mathcal{C})$, $(L_1,L_2)$ and $(L'_1,L'_2)$ are independent. Now,
\begin{align}
\ent(&L_1,L_2|L_0,L'_0,Q^n,\mathcal{C})   = n(\rho'_1+\tilde{\rho}'_1),
\label{HL1}\\
\ent(&L'_1,L'_2|L_0,L'_0,Q^n,\mathcal{C}) = n(\rho'_2+\tilde{\rho}'_2),
\label{HL2}\\
\ent(&T_1,S_1|Z^n,L_0,L'_0,Q^n,\mathcal{C}) \twocolbreak
\mathop  \leq  \limits^{(a)} n(\rho_1 + \rho_2 -\mi(U_1,V_1;Z|Q,U_0,V_0)+\epsilon),
\label{HL3}\\
\ent(&T_2,S_2|Z^n,L_0,L'_0,Q^n,\mathcal{C}) \twocolbreak
\mathop  \leq \limits^{(b)} n(\tilde{\rho}_1 + \tilde{\rho}_2 -\mi(U_2,V_2;Z|Q,U_0,V_0)+\epsilon),
\label{HL4}
\end{align}
where $(a)$ is due to the following. Consider,
\begin{align}
\ent(&T_1,S_1|Z^n,L_0,L'_0,Q^n,\mathcal{C}) \twocolbreak
= \ent(T_1,S_1|U_0^n(L_0),V_0^n(L'_0),Z^n,L_0,L'_0,Q^n,\mathcal{C})\nonumber\\
& \leq \ent(T_1,S_1|U_0^n(L_0),V_0^n(L'_0),Z^n,Q^n,\mathcal{C})\nonumber.
\end{align}
\sloppy Now we upper bound the term
$\ent(T_1,S_1|U_0^n(L_0),V_0^n(L'_0),Z^n,Q^n,\mathcal{C})$. From
\eqref{typicality} we have $\Prob\big((Q^n,U_0^n(L_0),U_1^n(L_0,t_1(L_0,L_1)),V_0^n(L'_0)$, $V_1^n(L'_0,s_1(L'_0,L'_1)),Z^n)\in\mathcal{T}_{\epsilon}^{(n)}\big)\to 1$ as $n\to \infty$. Applying Lemma~\ref{lemma1} leads to,
\begin{align}
&\ent(T_1,S_1|U_0^n(L_0),V_0^n(L'_0),Z^n,Q^n,\mathcal{C})\twocolbreak
\leq n(\rho_1 + \rho_2 -\mi(U_1,V_1;Z|Q,U_0,V_0)+\epsilon),
\end{align}if $\rho_1 \geq \mi(U_1;Z|Q,U_0,V_0)+\epsilon$, $\rho_2
\geq \mi(V_1;Z|Q,U_0,V_0)+\epsilon$, and $\rho_1 + \rho_2 \geq \mi(U_1,V_1;Z|Q,U_0,V_0)+\epsilon$. By the same argument the inequality $(b)$ holds, if the following inequalities hold,
\begin{align}
\tilde{\rho}_1 &\geq \mi(U_2;Z|Q,U_0,V_0)+\epsilon,\nonumber\\
\tilde{\rho}_2 &\geq \mi(V_2;Z|Q,U_0,V_0)+\epsilon,\nonumber\\
\tilde{\rho}_1 + \tilde{\rho}_2 &\geq \mi(U_2,V_2;Z|Q,U_0,V_0)+\epsilon,\nonumber
\end{align}
Substituting \eqref{HL1}-\eqref{HL4} into \eqref{IT} leads to,
\begin{align}
\mi(&T_1,T_2, S_1,S_2;Z^n|L_0,L'_0,Q^n,\mathcal{C})\nonumber\\
\geq{}& n(\rho'_1+\tilde{\rho}'_1) + n(\rho'_2+\tilde{\rho}'_2) \twocolbreak
 - n(\rho_1 + \rho_2 -\mi(U_1,V_1;Z|Q,U_0,V_0)+\epsilon) \nonumber\\
&- n(\tilde{\rho}_1 + \tilde{\rho}_2 - \mi(U_2,V_2;Z|Q,U_0,V_0)+\epsilon).
\label{fourthterm}
\end{align}
Substituting \eqref{firstterm}-\eqref{thidterm1} and \eqref{fourthterm} into \eqref{Equivocation} yields
\begin{align}
&\mi(W_1,W_2;Z^n|Q^n,\mathcal{C})\twocolbreak
\leq n\mi(U_0,U_1,U_2,V_0,V_1,V_2;Z|Q)  - n(\tilde{R}_1 - R_1 + \tilde{R}_2 - R_2)\nonumber\\
& + n(\tilde{R}_1 - R_1 + \tilde{R}_2 - R_2 -
\mi(U_0,V_0;Z|Q) )\twocolbreak
 - n(\rho'_1+\tilde{\rho}'_1) - n(\rho'_2+\tilde{\rho}'_2)  \nonumber\\
& + n(\rho_1 + \rho_2 -\mi(U_1,V_1;Z|Q,U_0,V_0)+\epsilon) \twocolbreak
+ n(\tilde{\rho}_1 + \tilde{\rho}_2 - \mi(U_2,V_2;Z|Q,U_0,V_0)+\epsilon).
\end{align}Therefore $\mi(W_1,W_2;Z^n|Q^n,\mathcal{C})\leq n\epsilon$ if
\begin{align}
&\mi(U_1,U_2,V_1,V_2;Z|U_0,V_0) \twocolbreak
- \rho'_1 - \tilde{\rho}'_1 - \rho'_2 - \tilde{\rho}'_2 + \rho_1 + \rho_2 - \mi(U_1,V_1;Z|Q,U_0,V_0)\nonumber\\
& + \tilde{\rho}_1 + \tilde{\rho}_2- \mi(U_2,V_2;Z|Q,U_0,V_0) \leq \epsilon.
\end{align}As a result, the rate constraints derived in equivocation analysis are
\begin{align}
&\tilde{R}_1 - R_1 + \tilde{R}_2 - R_2 \geq \mi(U_0,V_0;Z|Q), 
\label{equiv1}\\
&\tilde{R}_1 - R_1 \geq \mi(U_0;Z|Q), 
\label{equiv11}\\
&\tilde{R}_2 - R_2 \geq \mi(V_0;Z|Q), 
\label{equiv12}\\
&\rho_1 + \rho_2 \geq \mi(U_1,V_1;Z|Q,U_0,V_0),
\label{equiv2}\\
&\rho_1 \geq \mi(U_1;Z|Q,U_0,V_0),
\label{equiv3}\\
&\rho_2 \geq \mi(V_1;Z|Q,U_0,V_0),
\label{equiv4}\\
&\tilde{\rho}_1 + \tilde{\rho}_2 \geq \mi(U_2,V_2;Z|Q,U_0,V_0),
\label{equiv5}\\
&\tilde{\rho}_1 \geq \mi(U_2;Z|Q,U_0,V_0),
\label{equiv6}\\
&\tilde{\rho}_2 \geq \mi(V_2;Z|Q,U_0,V_0),
\label{equiv70}\\
&\rho_1 + \rho_2 + \tilde{\rho}_1 + \tilde{\rho}_2 - \rho'_1 - \tilde{\rho}'_1 - \rho'_2 - \tilde{\rho}'_2\twocolbreak
\leq \mi(U_1,V_1;Z|Q,U_0,V_0) + \mi(U_2,V_2;Z|Q,U_0,V_0) \nonumber\\
&\,\,\,\,\,- \mi(U_1,U_2,V_1,V_2;Z|U_0,V_0).
\label{equiv7}
\end{align}
Finally, by applying the Fourier-Motzkin procedure \cite{Shirani2011} to \eqref{dec01}, \eqref{dec02}, \eqref{eq:dec11}-\eqref{eq:dec25},  and \eqref{equiv1}-\eqref{equiv7} we obtain the inequalities in Theorem~\ref{thper}.
\end{appendices}

\section*{Acknowledgment}
The authors would like to thank anonymous reviewers of ISIT 2018 for their helpful comments.
\bibliographystyle{IEEEtran}
\bibliography{IEEEabrv,2T2R}

\end{document}